%% file: patch-locale-full.tex
\documentclass[a4paper,reqno]{amsart}

\PassOptionsToPackage{usenames,dvipsnames}{xcolor}
\usepackage{ifthen}
\usepackage{mathtools}
\usepackage{graphicx}
\usepackage{enumitem}
\usepackage{tikz-cd}
\usepackage{xcolor}

\newboolean{arxiv-version}
\setboolean{arxiv-version}{true}

\input{macros}

\ifthenelse{\boolean{arxiv-version}}{

\theoremstyle{plain}
\newtheorem{theorem}{Theorem}[section]
\newtheorem{lemma}[theorem]{Lemma}
\newtheorem{corollary}[theorem]{Lemma}
\theoremstyle{definition}
\newtheorem{definition}[theorem]{Definition}
\newtheorem{remark}[theorem]{Remark}
\newtheorem{convention}[theorem]{Convention}
\usepackage{natbib}
\usepackage{scalefnt}
\usepackage{hyperref}

\hypersetup{
  colorlinks = true,
  linktoc    = page,
  citecolor  = MidnightBlue,
  linkcolor  = black,
  pdfauthor  = {\IAName{}\ and\ \MHEName{}\ and \ATName{}},
  pdftitle   = {The Patch Topology in Univalent Foundations},
}

}{%

\usepackage{amsmath}
\usepackage{amsfonts}
\newtheorem{convention}[therm]{Convention}

}

\mathchardef\mhyphen="2D
\newcommand{\blank}{\mhyphen}

\begin{document}

\ifthenelse{\boolean{arxiv-version}}{

\title[The Patch Topology in Univalent Foundations]{The Patch Topology\\ in Univalent Foundations}

}{%

\lefttitle{The Patch Topology in Univalent Foundations}
\righttitle{Mathematical Structures in Computer Science}
\papertitle{Article}
\jnlPage{1}{00}
\jnlDoiYr{2024}
\doival{10.1017/xxxxx}
\title{The Patch Topology in Univalent Foundations}

}

\begin{abstract}
  Stone locales together with continuous maps form a coreflective subcategory of
  spectral locales and perfect maps. A proof in the internal language of an
  elementary topos was previously given by the second-named author. This proof can
  be easily translated to univalent type theory using \emph{resizing axioms}. In
  this work, we show how to achieve such a translation \emph{without} resizing
  axioms, by working with large and locally small frames with small bases. This
  requires predicative reformulations of several fundamental concepts of locale
  theory in predicative \VHoTTUF{}, which we investigate systematically.
\end{abstract}

\ifthenelse{\boolean{arxiv-version}}{

\author[\IALastName{}]{\IAName{}}
\author[\MHELastName{}]{\MHEName{}}
\author[\ATLastName{}]{\ATName{}}
}{%
\begin{authgrp}
\author{\IAName{}}
\affiliation{%
  University of Birmingham\\
  \email{i.arrietatorres@bham.ac.uk}%
}
\author{\MHEName{}}
\affiliation{%
  University of Birmingham\\
  \email{m.escardo@bham.ac.uk}
}%
\author{~\ATName{}}
\affiliation{%
  University of Birmingham\\
  \email{a.tosun@pgr.bham.ac.uk}%
}%
\end{authgrp}
\history{(Received xx xxx xxx; revised xx xxx xxx; accepted xx xxx xxx)}
}

\ifthenelse{\boolean{arxiv-version}}{}{%
}

\maketitle

\section{Introduction}\label{sec:intro}

The category $\Stone$ of Stone locales together with continuous maps forms a
coreflective subcategory of the category $\Spec$ of spectral locales and
\emph{perfect} maps. A proof in the internal
language of an elementary topos was previously constructed by
\cite{mhe-patch-short, mhe-patch-full}, defining the patch frame as the frame of
Scott continuous nuclei on a given frame.

In the present work, we show that the same construction can be carried out in
predicative and constructive univalent foundations. In the presence of
Voevodsky's~\citeyearpar{voevodsky-resizing} resizing axioms, it is
straightforward to translate this proof to univalent type theory. At the time of
writing, however, there is no known computational interpretation of the resizing
axioms. Therefore, the question of whether resizing can be avoided in the
construction of the patch locale is of interest.
In such a predicative situation, the usual approach to locale
theory is to work with presentations of locales, known as \emph{formal
topologies}~\FormalTopologyCitations{}. We show in our work, however, that it is
possible to work with locales directly. This requires a number of
modifications to the proofs and constructions of~\cite{mhe-patch-short,
mhe-patch-full}:
\begin{enumerate}
  \item The major modification is that we work with large and locally small
    frames with small bases. The reason for this is that, in the absence of
    propositional resizing, there are no non-trivial examples of
    frames~\citep{tdj-mhe-aspects}. See Section~\ref{subsec:definition-of-locale}
    for further discussion.
  \item The patch frame is defined as the frame of \VScottContinuous{} nuclei by
    \cite{mhe-patch-short, mhe-patch-full}. In order to prove that this is indeed a
    frame, one starts with the frame of all nuclei, and then exhibits the
    \VScottContinuous{} nuclei as a subframe. This procedure, however, does not seem
    to be possible in the context of our work, as it is not clear whether all nuclei
    can be shown to form a frame predicatively; so we construct the frame of
    \VScottContinuous{} nuclei \emph{directly}, which requires predicative
    reformulations of all proofs about
    it inherited from the frame of all nuclei.
  \item\label{item:aft} In the impredicative setting, any frame has all Heyting
    implications, which are needed to construct open nuclei. Again, this does not
    seem to be the case in the predicative setting. We show, however, that it is
    possible to construct Heyting implications in large and locally small frames
    with small bases, by an application of the posetal instance of the Adjoint
    Functor Theorem.
  \item Similar to (\ref{item:aft}), we use the posetal Adjoint Functor Theorem
    to define the right adjoint of a frame homomorphism, which we then use to
    define the notion of a \emph{perfect map}, namely, a map whose defining frame
    homomorphism's right adjoint is \VScottContinuous{}. This notion is used by
    \MHELastName{} in his impredicative proof
    \citeyearpar{mhe-patch-short, mhe-patch-full}.
 % \item In our proof of the universal property, we make use of the fact that the
 %   patch locale is locally small, which depends on the existence of a small basis
 %   on the spectral locale that we start with.
\end{enumerate}

For the purposes of our work, a \emph{spectral locale} is a locale in which any
open can be expressed as the join of a family of compact opens. A continuous map
of spectral locales is \emph{spectral} if its defining frame homomorphism
preserves compact opens. A \emph{Stone locale} is one that is compact and in
which any open is expressible as the join of a family of clopens. Every Stone
locale is spectral since the clopens coincide with the compact opens in Stone
locales. The patch frame construction is the right adjoint to the inclusion
$\Stone \hookrightarrow \Spec$. The main contribution of our work is the construction of this
right adjoint in the predicative context of univalent type theory. We have also
formalized~\citep{type-topology-locale-theory} the development of this paper in
the \VAgda{} proof assistant~\citep{norell-agda} as part of the \VTypeTopology{}
library~\citep{type-topology}. Our presentation here, however, is self-contained
and can be followed independently of the formalization.

The organization of this paper is as follows. In Section~\ref{sec:foundations},
we present the type-theoretical context in which we work. In
Section~\ref{subsec:definition-of-locale}, we introduce our notion of locale
discussed above. In Section~\ref{sec:spec-and-stone}, we present our definitions
of spectral and Stone locales that provide a suitable basis for a predicative
development. In Section~\ref{sec:aft}, we present the posetal instance of the
Adjoint Functor Theorem for the simplified context of locales that is central to
our development. In Section~\ref{sec:meet-semilattice}, we define the
meet-semilattice of perfect nuclei as preparation for the complete lattice of
perfect nuclei, which we then construct in Section~\ref{sec:joins}. Finally in
Section~\ref{sec:coreflection}, we prove the desired universal property, namely,
that the patch locale exhibits the category $\Stone$ as a coreflective
subcategory of $\Spec$, where we restrict ourselves to locales with small bases.

Finally, we note that a preliminary version of the work that we present here
previously appeared in \citep{patch-mfps}. Our work here extends
\emph{loc.~cit.}\ in several directions. The presentation that we provide
through Theorem~\ref{thm:answer} and Theorem~\ref{thm:answer-stone} here
involves new results. Furthermore, we provide a complete proof of the universal
property of $\Patch$ through a new approach, as given in
Theorem~\ref{thm:patch-is-stone} and Theorem~\ref{thm:main}.

\section{Foundations}
\label{sec:foundations}

In this section, we introduce the type-theoretical setting in which we work and
then present the type-theoretical formulations of some of the preliminary
notions that form the basis of our work. Our type-theoretical conventions follow
those of de~Jong and \MHELastName{}~\citeyearpar{tdj-mhe-aspects, tdj-thesis}
and \cite{ufp-hottbook-2013}.

We work in Martin-L\"of Type Theory with binary sums $(\blank) + (\blank)$,
dependent products $\PiTySym$, dependent sums $\SigmaTypeSym$, the identity type
$(\blank) = (\blank)$, and inductive types including the empty type $\EmptyTy$,
the unit type $\UnitTy$, and the type $\ListTy{A}$ of lists, or \emph{words},
over any type $A$. We denote by $\projI$ and $\projII$ the first and the second
projections of a $\SigmaTypeSym$ type. We adhere to the convention of the HoTT
Book~\citeyearpar{ufp-hottbook-2013} of using $(\blank) \equiv (\blank)$ for
judgemental equality and $(\blank) = (\blank)$ for the identity type.

We work explicitly with universes, for which we adopt the convention of using
the variables $\UU, \VV, \WW, \ldots$ The ground universe is denoted $\UU_0$ and the
successor of a given universe $\UU$ is denoted $\USucc{\UU}$. The least upper
bound of two universes is given by the operator $(\blank) \sqcup (\blank)$ which is
assumed to be associative, commutative, and idempotent. We do not assume that
the universes are, or are not, cumulative. Furthermore, $\USucc{(\blank)}$ is
assumed to distribute over $(\blank) \sqcup (\blank)$. Universes are computed for the
given type formers as follows:

\begin{itemize}
  \item Given types $X : \UU$ and $Y : \VV$, the type $X + Y$ inhabits universe
    $\UU \sqcup \VV$.
  \item Given a type $X : \UU$ and an $X$-indexed family, $Y : X \to \VV$, both
    $\sum_{x : X}Y(x)$ and $\prod_{x : X}Y(x)$ inhabit the universe $\UU \sqcup \VV$.
  \item Given a type $X : \UU$ and inhabitants $x, y : X$, the identity type $x
    = y$ inhabits universe $\UU$.
  \item The type $\NatTy$ of natural numbers inhabits $\UU_0$.
  \item The empty type $\EmptyTy$ and the unit type $\UnitTy$ have copies in
    every universe $\UU$, which we occasionally make explicit using the
    notations $\EmptyTy_{\UU}$ and $\UnitTy_{\UU}$.
  \item Given a type $A : \UU$, the type $\ListTy{A}$ inhabits $\UU$.
\end{itemize}

We assume the univalence axiom and therefore function extensionality and
propositional extensionality. We maintain a careful distinction between
\emph{structure} and \emph{property}, and reserve logical connectives for
\define{propositional types} i.e.\ types $A$ satisfying $\IsProp{A} \is \PiTy{x,
y}{A}{x = y}$. We denote by $\HProp{\UU}$ the type of propositional types in
universe $\UU$ i.e.\ $\HProp{\UU} \is \SigmaType{A}{\UU}{\IsProp{A}}$. A type
$A$ is called a \define{set} if its identity type is always a proposition
i.e.\ $\IdTy{x}{y}$ is a propositional type, for every $x, y : A$.

We assume the existence of \define{propositional truncation}, given by a type
former $\TruncTy{\blank} : \UU \to \UU$ and a unit operation $\TruncTm{\blank} : A
\to \TruncTy{A}$. The existential quantification operator is defined using
propositional truncation as
$\exists x : A, B(x) \is \TruncTy{\Sigma x : A, B(x)}$.

When presenting proofs informally, we adopt the following conventions for
avoiding ambiguity between propositional and non-propositional types.
\begin{itemize}
  \item For the anonymous inhabitation $\TruncTy{A}$ of a type, we say that $A$
    is inhabited.
  \item For truncated $\SigmaTypeSym$ types, we use the terminologies
    \emph{there is} and \emph{there exists}.
  \item We say \emph{specified inhabitant} of type $A$ to contrast it with the
    anonymous inhabitation $\TruncTy{A}$. Similarly, we say there is a
    \emph{specified} or \emph{chosen} element to emphasize that we are using
    $\SigmaTypeSym$ instead of $\bigexists$.
  \item When we use the phrase \emph{has a}, we take it to mean a
    \emph{specified} inhabitant. In contrast, we say \emph{has some} when talking
    about an unspecified inhabitant. If we want to completely avoid ambiguity,
    we prefer to use the more explicit terminology of
    \emph{has specified} or \emph{has unspecified}.
\end{itemize}

If a given type $A$ is a proposition, it is clear that it is logically
equivalent to its own propositional truncation. The converse, however, is not
always true: there are types that are logically equivalent to their own
propositional truncations, despite not being propositions themselves. Types that
satisfy this more general condition of being logically equivalent to their own
truncations are said to have \emph{split support}, and have previously been
investigated by~\cite{keca-anonymous}.

%\subsection{Notions of size}\label{subsec:size}

%We start by defining the notion of $\VV$-smallness~\citep{tdj-mhe-aspects}.

\begin{definition}[$\VV$-small type]\label{defn:vsmall}
  A type $A : \UU$ is called \define{$\VV$-small} if it has a copy in universe
  $\VV$. That is to say, $\SigmaType{B}{\VV}{\Equiv{A}{B}}$.
\end{definition}

\begin{lemma}\label{lem:being-small-is-prop} The following are equivalent.
  \begin{enumerate}
    \item For every type $A : \UU$ and universe $\VV$, the property of being
      \VVSmall{\VV}\ is a proposition.
    \item Univalence holds.
  \end{enumerate}
\end{lemma}

We also have the notion of \emph{local smallness} which refers to the equality
type always being small with respect to some universe.

\begin{definition}[Local $\VV$-smallness]\label{defn:local-smallness}
  A type $X : \UU$ is called \define{locally $\VV$-small} if the identity
  type $\IdTy{x}{y}$ is \VVSmall{\VV} for every $x, y : X$.
\end{definition}

We have mentioned the \emph{propositional resizing} axiom in
Section~\ref{sec:intro}. Having formally defined the notion of $\VV$-smallness,
we can now give the precise definition of this axiom.

\begin{definition}
  The \define{propositional $(\UU, \VV)$-resizing axiom} says that any
  proposition $P : \HProp{\UU}$ is \VVSmall{\VV}.
  The \define{global propositional resizing} axiom says that, for any two
  universes $\UU$ and $\VV$, the propositional
  \VResizing{\UU}{\VV} axiom holds.
\end{definition}

\begin{definition}[The set replacement axiom]\label{defn:sr}
  Let $f : X \to Y$ be a function from a type $X : \UU$ into a set $Y : \VV$.
  The \define{set replacement axiom} says that if $X$ is \VVSmall{\UU'} and the
  type $Y$ is locally $\WW$-small then $\image{f}{X}$ is \VVSmall{(\UU' \lub
  \WW)}, where
  \begin{equation*}
    \image{f}{X} \is \SigmaType{y}{Y}{\bigexists_{x : X} \IdTy{f(x)}{y}}.
  \end{equation*}
\end{definition}

\begin{remark}\label{rmk:sr-special-case}
  Notice that we obtain the following as a special case of the set replacement
  axiom: given a locally $\UU$-small set $X' : \UU^+$, the image of any function
  $f : X \to X'$ is $\UU$-small whenever the type $X$ is $\UU$-small.
\end{remark}

%\subsection{Directed families}

We now proceed to define, in the presented type-theoretical setting, some
preliminary notions that are fundamental to our development of locale theory.

\begin{definition}[Family]\label{defn:family}
  A \define{$\UU$-family on a type} $A$ is a pair $\Pair{I}{\alpha}$
  where $I : \UU$ and $\alpha : I \to A$.
  We denote the type of $\UU$-families on type $A$ by $\Fam{\UU}{A}$
  i.e.\ $\Fam{\UU}{A} \is \SigmaType{I}{\UU}{I \to A}$.
\end{definition}

\begin{convention}
  We often use the shorthand $\FamEnum{i}{I}{x}$ for families. In other words,
  instead of writing $(I, x)$ for a family, we write $\FamEnum{i}{I}{x_i}$ where
  $x_i$ denotes the application $x(i)$. Given a family $(I, \alpha)$, a subfamily
  of it is a family of the form $(J, \alpha \circ \beta)$ where $(J, \beta)$ is a family on the
  index type $I$. When talking about a subfamily of some family
  $\FamEnum{i}{I}{x_i}$, we use the notation $\FamEnum{j}{J}{x_{i_j}}$ to denote
  the subfamily given by a family $\FamEnum{j}{J}{i_j}$
\end{convention}

We will also be talking about families $(I, \alpha)$ where the function $\alpha$ is an
\emph{embedding}. We refer to such families as \emph{embedding families}.

\begin{definition}[Embedding]\label{defn:embedding}
  A function $f : X \to Y$ is called an \define{embedding} if, for every $y :
  Y$, the type $\SigmaType{x}{X}{\IdTy{f(x)}{y}}$ is a proposition.
\end{definition}

\begin{definition}\label{defn:emb-family}
  A $\UU$-family $(I, \alpha)$ on a type $A$ is said to be an \define{embedding
  family} if the function $\alpha : I \to A$ is an \VEmbedding{}. We denote the type of
  \define{embedding families} by $\EmbFam{\UU}{A}$.
\end{definition}

\section{Locales} \label{subsec:definition-of-locale}

A \emph{locale} is a notion of space characterized solely by its lattice of
opens. The lattice-theoretic notion abstracting the behaviour of a lattice of
open subsets is a \emph{frame}: a lattice with finite meets and arbitrary joins
in which the binary meets distribute over arbitrary joins.

Our type-theoretic definition of a frame is parameterized by three universes:
(1) for the carrier set, (2) for the order, and (3) for the join-completeness
i.e.\ the index types of families on which the join operation is defined. We
adopt the convention of using the universe variables $\UU$, $\VV$, and $\WW$ for
these respectively.

\begin{definition}[Frame]\label{defn:frame}
  A \define{$(\UU, \VV, \WW)$-frame} $L$ consists of:
  \begin{itemize}
    \item a set $| L | : \UU$,
    \item a partial order $(\blank) \le (\blank) : | L | \to | L | \to \Omega_\VV$,
    \item a top element $\OneSym : | L |$,
    \item an operation $(\blank) \wedge (\blank) : | L | \to | L | \to | L |$ giving the
      greatest lower bound $U \wedge V$ of any two $U, V : | L |$,
    \item an operation $\bigvee(\blank) : \Fam{\WW}{| L |} \to | L |$ giving the least upper
      bound $\bigvee (I, \alpha)$ of any $\WW$-family $(I, \alpha)$,
  \end{itemize}
  such that binary meets distribute over arbitrary joins, i.e.
  $U \wedge \bigvee_{i : I} x_{i} = \bigvee_{i : I} U \wedge x_{i}$
  for every $x : | L |$ and $\WW$-family $\FamEnum{i}{I}{x_i}$ in $| L |$.
\end{definition}
As we have done in the definition above, we adopt the shorthand notation
$\bigvee_{i : I} x_i$ for the join of a family $\FamEnum{i}{I}{x_i}$.
We also adopt the usual abuse of notation and write $L$ instead of $| L |$.

The reader might have noticed that we have not imposed a sethood condition on
the underlying type of a frame in this definition. The reason for this is that
it follows automatically from the antisymmetry condition for partial orders that
the underlying type of a frame is a set.

% Our use of the term \emph{locally small} here is ostensibly different from the
% one that we gave in Definition~\ref{defn:local-smallness}. Provided that the
% order in consideration is antisymmetric, which it \emph{is} in the case of a
% frame, the two notions in fact coincide and can hence be used interchangeably.

\begin{lemma}\label{prop:local-smallness-equiv}
  Let $L$ be a $(\UU^+, \VV, \UU)$-frame and let $x, y : L$. We have that $x \le
  y$ is \VVSmall{\UU} if and only if the carrier of $L$ is a locally $\UU$-small type.
\end{lemma}
\begin{proof}
  Let $x, y : L$. It is a standard fact of lattice theory that $x \le y \leftrightarrow \IdTy{x
  \wedge y}{x}$. ($\Rightarrow$) If the frame is locally small in the sense that $x \le y$ is
  \VVSmall{\UU} for every $x, y : L$, then the identity type $\IdTy{x}{y}$ must
  also be \VVSmall{\UU} since $\IdTy{x}{y}$ is equivalent to the conjunction $(x \le
  y) \wedge (y \le x)$. ($\Leftarrow$) Conversely, if the carrier of $L$ is a locally small
  type, then $x \le y$ must be \VVSmall{\UU} since it is equivalent to the identity
  type $\IdTy{x \wedge y}{x}$ which is \VVSmall{\UU} by the local smallness of the
  carrier.
\end{proof}

We also note that, in the work of \TDJLastName{} and
\MHELastName{}~\citeyearpar{tdj-mhe-aspects, tdj-mhe-dt, tdj-thesis}, the term
\emph{$\VV$}-dcpo is used for a directed-complete partially ordered set whose
directed joins are over $\VV$-families. This terminology leaves the carrier and
the order implicit as in most cases, the completeness universe is the only one
relevant to the discussion.

We gave a highly general definition of the notion of frame in
Definition~\ref{defn:frame}: all three universes involved in the definition are
permitted to live in separate universes. This generality, however, is never
needed in our work. In fact, we know that locales of a certain form cannot be
constructed predicatively; it follows from a result due to
\cite{tdj-mhe-aspects} that the existence of a nontrivial small locale is
equivalent to a form of propositional resizing that is known to be independent
of type theory --- see~\cite[Theorem 35]{tdj-mhe-aspects} for details. In short,
the carrier set is naturally forced to be large when taking a predicative
approach to locale theory.

In light of this, we restrict attention to $(\UU^+, \UU, \UU)$-frames, which we
refer to as \emph{large and locally small} frames.

\begin{convention}
  From now on, we fix a \emph{base} universe $\UU$, and refer to types that have
  isomorphic copies in $\UU$ as \emph{small types}. In contrast, we refer to types
  in $\USucc{\UU}$ (with no isomorphic copies in $\UU$) as \emph{large types}.
  Accordingly, we hereafter take frame/locale to mean one whose underlying type is
  \emph{large and locally small} i.e.\ a $(\USucc{\UU}, \UU, \UU)$-frame/locale.
\end{convention}

% We now proceed to define the notion of frame homomorphism.

\begin{definition}[Frame homomorphism]\label{defn:frame-homomorphism}
  Let $K$ and $L$ be two locales.
  A function $h : |K| \to |L|$ is called a \define{frame homomorphism} if it
  preserves the top element, binary meets, and joins of small families. We denote
  by $\Frm$ the category of frames and their homomorphisms (over our base universe
  $\UU$).
\end{definition}

We adopt the notational conventions of \cite{sheaves}. A \define{locale} is a
frame considered in the opposite category denoted $\Loc \is \opposite{\Frm}$. To
highlight this, we adopt the standard conventions of (1) using the letters $X,
Y, Z, \ldots$ (or sometimes $A, B, C, \ldots$) for locales and (2) denoting the frame
corresponding to a locale $X$ by $\opens{X}$. For variables that range over the
frame of opens of a locale $X$, we use the letters $U, V, W, \ldots$ We use the
letters $f$ and $g$ for continuous maps $X \to Y$ of locales. A continuous map $f
: X \to Y$ is given by a frame homomorphism $f^* : \opens{Y} \rightarrow \opens{X}$.

\begin{definition}[Nucleus]\label{defn:nucleus}
  A \define{nucleus on a locale} $X$ is an endofunction $j : \opens{X} \to
  \opens{X}$ that is inflationary, idempotent, and preserves binary meets.
\end{definition}

In Section~\ref{sec:joins}, we will work with inflationary and
binary-meet-preserving functions that are not necessarily idempotent. Such
functions are called \define{prenuclei}. We also note that, to show a prenucleus
$j$ to be idempotent, it suffices to show $j(j(U)) \le j(U)$ as the other
direction follows from inflationarity. In fact, the notion of a nucleus could be
defined as a prenucleus satisfying the inequality $j(j(U)) \le j(U)$, but we
define it as in Definition~\ref{defn:nucleus} for the sake of simplicity and
make implicit use of this fact in our proofs of idempotency.

% \subsection{Bases of locales}

In point-set topology, a \emph{basis} for a space $X$ is a collection $\mathcal{B}$ of
subsets of $X$ such that every open $U \in \opens{X}$ can be decomposed into a
union of subsets contained in $\mathcal{B}$. These subsets are called the \emph{basic
opens}. In the context of pointfree topology, the notion of a basis is captured
by the lattice-theoretic notion of a \emph{generating lattice}. In
\citep[p.~548]{elephant-vol-1}, for example, a basis for a locale $\opens{X}$ is
defined as a subset $\mathcal{B} \subseteq \opens{X}$ such that every element of $\opens{X}$ is
expressible as a join of members of $\mathcal{B}$. We now give the formal definition of
this notion in our type-theoretical setting.

\begin{definition}[Directed family]\label{defn:directed}
  Let $\FamEnum{i}{I}{x_i}$ be a family on some type $A$ that is equipped with a
  preorder $(\blank) \le (\blank)$. This family is called \define{directed} if
  \begin{enumerate}
    \item $I$ is inhabited, and
    \item for every $i, j : I$, there exists some $k : I$ such that $x_k$ is an
      upper bound of $\{ x_i, x_j \}$.
  \end{enumerate}
\end{definition}

\begin{definition}[Basis of a frame]\label{defn:ext-basis}\label{defn:int-basis}
  Let $X$ be a locale. A \define{basis} for $X$ is an embedding family
  $\FamEnum{i}{I}{B_i}$
  such that for every $U : \opens{X}$, there is a specified directed small
  family $\FamEnum{j}{J}{i_j}$ on the basis $I$ satisfying
  $U = \bigvee_{j : J} B_{i_j}$.
  We often drop the embedding requirement, in which case we speak of
  \emph{intensional bases}. When we want to emphasize the distinction, we speak
  of \emph{extensional bases} when we require the embedding condition.
\end{definition}

Notice that the embedding requirement says that the family in consideration does
not have repetitions. However, it will often be convenient to work with families
which might have repetitions. Such repetitions can be removed (constructively)
if necessary, by factoring the family through its image. In the case of a
\VSpectralLocale{}, for example, there can be different intensional bases but
the extensional basis is unique.

Even though we do not require a basis to be given by a small family, we always
work with small bases in our work.

\begin{remark}\label{rmk:directification}
  In the definition above, we required covering families to be directed. This
  requirement is not essential since a basis can always be \emph{directified} by
  closing it under finite joins, that is, given a basis $\FamEnum{i}{I}{B_i}$,
  we can define $B^{\uparrow} : \ListTy{I} \to \opens{X}$ as
  \begin{equation*}
    B^{\uparrow}(i_0, \ldots, i_{n-1}) \quad\is\quad B_{i_0} \vee \cdots \vee B_{i_{n-1}},
  \end{equation*}
  which yields an equivalent basis. We nevertheless prefer to include the
  directedness requirement in the definition for technical convenience. Finally,
  we also note that the directification of a small basis gives a small basis.
\end{remark}

\begin{lemma}\label{prop:img-small}
  Given any locale $X$ with some small intensional basis $\FamEnum{i}{I}{B_i}$,
  the image of $B : I \to \opens{X}$ is small.
\end{lemma}
\begin{proof}
  Since being small is a propositional type (by
  Lemma~\ref{lem:being-small-is-prop}), we can appeal to the induction
  principle of propositional truncation and assume we are given a \emph{specified}
  intensional small basis $\FamEnum{i}{I}{B_i}$.
  By appealing to the set replacement axiom (Definition~\ref{defn:sr}) as
  explained in Remark~\ref{rmk:sr-special-case}, it just remains to show that the
  carrier of $\opens{X}$ is locally small, which is always the case in a locally
  small locale thanks to Lemma~\ref{prop:local-smallness-equiv}.
\end{proof}

\section{Spectral and Stone Locales}
\label{sec:spec-and-stone}

The standard impredicative definition of a spectral locale is as one in which
the compact opens form a basis closed under finite meets (see \RTwo{}.3.2 from
\citep{ptj-ss}). To talk about compactness, we define the \emph{way-below}
relation.

\begin{definition}[Way-below relation]\label{defn:way-below}
  We say that an open $U$ of a locale $X$ is \define{way~below} an open $V$,
  written $U \WayBelow V$,
  if for every directed family $(W_i)_{i : I}$ with $V \le \bigvee_{i : I} W_{i}$
  there is some $i : I$ with $U \le W_i$.
\end{definition}

\begin{lemma}\label{prop:way-below-prop}
  Given any two opens $U, V : \opens{X}$, the type $U \WayBelow V$ is a
  proposition.
\end{lemma}

The statement $U \WayBelow V$ is thought of as expressing that open $U$ is
compact \emph{relative to} open $V$. An open is said to be compact if it is
compact relative to itself:

\begin{definition}[Compact open of a locale]
  An open $U : \opens{X}$ is called \define{compact} if $U \WayBelow U$.
\end{definition}

\subsection{Spectral locales}\label{sec:spectral-locales}

\begin{definition}[Compact locale]\label{defn:compact-locale}
  A locale $X$ is called \define{compact} if the top open $\One{X}$ is
  \VCompact{}.
\end{definition}

We denote by $\CompactOpens{X}$ the type of compact opens of a locale $X$. In
other words, we define
\begin{equation*}\label{defn:type-of-compact-opens}
  \CompactOpens{X} \quad\is\quad \SigmaType{U}{\opens{X}}{\IsCompact{U}}.
\end{equation*}
Since our locales are large and locally small, we have that the type of compact
opens lives in $\UU^+$ i.e.\ is not \emph{a priori} small. A spectral locale,
also sometimes called \emph{coherent} in the
literature~\citep[\RTwo{}.3.2]{ptj-ss}, is defined in impredicative locale
theory as one in which the compact opens (1) are closed under finite meets, and
(2) form a basis. In our formulation of this notion, we capture the same idea
but impose the additional requirement that the type $\CompactOpens{X}$ of
compact opens forms a small type.

\begin{definition}[Spectral locale]\label{defn:spectral-locale}
  A locale $X$ is called \define{spectral} if it satisfies the following
  conditions.
  \begin{enumerate}[label={$\left({\mathsf{SP{\arabic*}}}\right)$}, leftmargin=4em]
    \item It is a \VCompactLocale{} (i.e.\ the empty meet is compact).
      \label{item:sl-compact}
    \item The meet of two \VCompactOpen{}s is a \VCompactOpen.
      \label{item:sl-coherent}
    \item For any $U : \opens{X}$, there exists a small \VDirected{} family
      $\FamEnum{i}{I}{K_i}$, with each $K_i$ \VCompact, such that
      $U = \bigvee_{i : I} K_i$.
      \label{item:sl-covering}
    \item The type $\CompactOpens{X}$ is \VSmall{}.
    \label{item:sl-smallness}
  \end{enumerate}
\end{definition}

The last two conditions are together equivalent to saying that the compact opens
form a small basis, although this is not immediate (Theorem~\ref{thm:answer}).
Notice that in the impredicative notion of \VSpectralLocale{}, everything is
small. One way of understanding the above definition is that we only make the
\emph{carrier} large, but everything else, including the basis, remains small.
This is important for a variety of reasons. Without the smallness condition, it
does not seem possible to (1)~define Heyting implication, and hence open nuclei,
(2)~show that frame homomorphisms have right adjoints, and so define the notion
of perfect frame homomorphism. Additionally, a natural Stone-type duality to
consider is that between small distributive lattices and spectral locales, for
which it becomes absolutely necessary that we work with small bases as in the
above definition. Notice that the same kind of phenomenon occurs when we work
with locally presentable categories, which are large, locally small categories
with a small set of objects that suitably generate all objects by taking
filtered colimits (\cite{adamek_rosicky_1994,elephant-vol-1}).

\begin{lemma}\label{prop:spectral-is-prop}
  For any locale $X$, the statement that $X$ is spectral is a proposition.
\end{lemma}
\begin{proof}
  Immediate from Lemma~\ref{prop:way-below-prop} and the fact that
  the $\PiTySym$ type over a family of propositions is a proposition.
\end{proof}

The natural notion of morphism between spectral locales is that of a
\emph{spectral map}.

\begin{definition}[Spectral map]\label{defn:spectral-map}
  A continuous map $f : X \to Y$ of spectral locales $X$ and $Y$ is called
  \define{spectral} if it reflects compact opens, that is, the open $f^*(V) :
  \opens{X}$ is compact of $X$ whenever the open $V$ is compact.
\end{definition}

We denote by $\Spec$ the category of spectral locales with spectral maps as the
morphisms. We recall the following useful fact (see e.g.\ \cite{mhe-patch-short}).
%
% We also note that the following well-known fact about spectral locales
% is valid under the definition of spectrality
% that we gave in Definition~\ref{defn:spectral-locale}.

\begin{lemma}\label{prop:spectral-yoneda}
  Let $X$ be a spectral locale and let $U, V : \opens{X}$ be two
  opens.
  \begin{center}
    $U \le V$ \quad if and only if \quad $\PiTy{K}{\opens{X}}{\IsCompact{K} \to K \le U \to K \le V}$.
  \end{center}
\end{lemma}

In the context of the definition of spectrality that we gave in
Definition~\ref{defn:spectral-locale}, where we have conditions ensuring that
the compact opens form a small basis, it is natural to ask whether the
basis in consideration unique. This is indeed the case, but it is
subtler than one might expect, and is proved in Theorem~\ref{thm:answer} below. We
first need some preparation.

Given that the conditions in Definition~\ref{defn:spectral-locale} capture the
idea of the compact opens behaving like a small basis closed under finite meets,
one might wonder if the same notion can be formulated by starting with a small
basis closed under finite meets and requiring it to consist of compact opens.

% We
% will call this notion \emph{spectral basis} and its intensional form
% \emph{intensional spectral basis}, following the terminology of
% Definition~\ref{defn:int-basis}.

\begin{definition}[Intensional spectral basis]\label{defn:int-spec-basis}
  An \define{intensional spectral basis} for a locale $X$ is a small
  intensional basis $\FamEnum{i}{I}{B_i}$ for the frame $\opens{X}$
  satisfying the following three conditions.
  \begin{enumerate}
    \item For every $i : I$, the open $B_i$ is \VCompact{}.
    \item There is an index $t : B$ such that $B_t = \One{X}$.
    \item For any two $i, j : I$, there exists some $k : I$ such that
      $B_k = B_i \wedge B_j$.
  \end{enumerate}
  We say that the basis is \emph{extensional} if
  the map $B : I \to \opens{X}$ is an \VEmbedding{}.
  When we say \emph{spectral basis} without any qualification,
  we mean an extensional spectral basis.
\end{definition}

In Remark~\ref{rmk:directification}, we mentioned that a basis can always be
directified by closing it under finite joins. The same applies to spectral bases
as well since the join of a finite family of compact opens is compact.

\begin{lemma}\label{lem:cmp-bsc}
  Given any locale $X$ with an intensional basis $\FamEnum{i}{I}{B_i}$, every
  compact open of $X$ falls in the basis, that is, for every
  \VCompactOpen{} $K$, there exists an index $i : I$ such that $B_i = K$.
\end{lemma}
\begin{proof}
  Let $K : \opens{X}$ be a compact open. As $\FamEnum{i}{I}{B_i}$ is an
  intensional basis, there must be a specified directed family
  $\FamEnum{j}{J}{i_j}$ on $I$ such that
  $K = \bigvee_{j : J} B_{i_j}$. By the compactness of $K$, there must be some
  $k : J$ such that $K \le B_{i_k}$.
  Clearly, $B_{i_k} \le K$ is also the case, since $K$ is an upper bound of the
  family $\FamEnum{j}{J}{B_{i_j}}$, which means we have that
  $\IdTy{K}{B_{i_k}}$.
\end{proof}

\begin{corollary}\label{prop:basis-img-equiv}
  Given any locale $X$ with an intensional basis $\FamEnum{i}{I}{B_i}$
  consisting of compact opens, we have an equivalence of types
  $\Equiv{\image{B}{I}}{\CompactOpens{X}}$.
\end{corollary}

\begin{lemma}\label{lem:basis-to-spec}
  If a locale $X$ has an unspecified, intensional, spectral, and small basis,
  then $X$ is \VSpectral{}.
\end{lemma}
\begin{proof}
  First, notice that the conclusion is a proposition since being spectral is a
  proposition (by Lemma~\ref{prop:spectral-is-prop}). This means that we
  may appeal to the induction principle of propositional truncation and assume we
  have a specified intensional spectral basis $\FamEnum{i}{I}{B_i}$. We know that
  $\Equiv{\image{B}{I}}{\CompactOpens{X}}$ by
  Lemma~\ref{prop:basis-img-equiv}, which is to say that the type
  $\CompactOpens{X}$ is small, by the smallness of $\image{B}{I}$
  (given by Lemma~\ref{prop:img-small}) meaning \ref{item:sl-smallness}
  is satisfied.

  The top element is basic and hence compact so \ref{item:sl-compact} holds.
  Given two compact opens $K_1$ and $K_2$, they must be basic meaning there exist
  $k_1, k_2$ such that $K_1 = B_{k_1}$ and $K_2 = B_{k_2}$. Because the basis is
  closed under binary meets there must be some $k_3$ with $B_{k_3} = K_1 \wedge K_2$
  which means condition \ref{item:sl-coherent} also holds.

  For \ref{item:sl-covering}, consider an open $U : \opens{X}$. We know that
  there is a specified a small family $\FamEnum{j}{J}{i_j}$ of indices such that
  $U = \bigvee_{j : J} B_{i_j}$.
  The subfamily $\FamEnum{j}{J}{B_{i_j}}$ is then clearly a small directed
  family with each $B_{i_j}$ \VCompact{} which is what we needed.
\end{proof}

\begin{lemma}\label{lem:spec-gives-basis}
  If a locale $X$ is \VSpectral{}, then it has a specified, extensional, and
  spectral small basis.
\end{lemma}
\begin{proof}
  Let $X$ be a \VSpectral{} locale. We claim that the inclusion $\CompactOpens{X} \hookrightarrow \opens{X}$, which is formally given by the first projection, is an extensional and spectral small basis. The fact that it is small
  is given by \ref{item:sl-smallness}. By \ref{item:sl-compact} and
  \ref{item:sl-coherent}, we know that this basis contains $\One{X}$ and is closed
  under binary meets. It remains to show that it forms a basis.
  To define a covering family for an open $U : \opens{X}$, we let the index type be the subtype of compact opens below $U$, which is again small by \ref{item:sl-smallness}, and the family is again given by the first projection.
  It is clear that $U$ is an upper bound of this family so it remains to show
  that it is its least upper bound. Consider some $V$ that is an upper bound
  of this family. By Lemma~\ref{prop:spectral-yoneda}, it suffices to show
  $K \le U$ implies $ K \le V$, for every compact open $K$. Any such compact open
  $K \le U$ is below $U$ by construction, which implies $K \le V$ since it
  is an upper bound.
\end{proof}

\begin{definition}[Extensionalization of an intensional basis] \label{def:extensionalization}
  For an intensional basis $\FamEnum{i}{I}{B_i}$ on a locale $X$, its
  \define{extensionalization}, is defined by taking the index set to be the
  $\image{B}{I}$, and the family to be the corestriction of $B$ to its image,
  which is given by the first projection, and hence is an embedding.
\end{definition}

\begin{lemma}
 If $\FamEnum{i}{I}{B_i}$ is a spectral basis, then so is its extensionalization.
\end{lemma}

\begin{lemma}\label{lem:split-support}
  From an unspecified intensional spectral basis on a locale $X$ we can obtain
  a specified extensional spectral basis.
\end{lemma}
\begin{proof}
  Let $X$ be a locale with an unspecified intensional spectral basis.
  By Lemma~\ref{lem:basis-to-spec}, we know that it is spectral and therefore
  that it has a specified extensional basis by Lemma~\ref{lem:spec-gives-basis}.
\end{proof}

Recall that a type $X$ is said to have \define{split support} if $\TruncTy{X} \to
X$~\citep{keca-anonymous}. With this terminology, the above lemma says that the
type of intensional spectral bases on a locale $X$ has \VSplitSupport{}.

\begin{theorem}\label{thm:answer}
The following are logically equivalent for any locale $X$.
\begin{enumerate}
  \item\label{it:I}%
    $X$ is \VSpectral.
  \item\label{it:II}%
    $X$ has an unspecified intensional small spectral basis.
  \item\label{it:IIPrime}%
    $X$ has an unspecified extensional small spectral basis.
 \item\label{it:V}%
    The inclusion $\CompactOpens{X} \hookrightarrow X$ is an extensional spectral basis,
    where to an open $U$ we assign the directed family of compact opens below it as in the construction of Lemma~\ref{lem:spec-gives-basis}.
  \item\label{it:III}%
    $X$ has a specified intensional small spectral basis.
  \item\label{it:IV}%
    $X$ has a specified extensional small spectral basis.
\end{enumerate}
\end{theorem}
Notice that the first four conditions are propositions, but the last two are not in general.
\begin{proof}
  We have already established one direction of the logical equivalence $(\ref{it:II}) \leftrightarrow (\ref{it:III})$ in
  Lemma~\ref{lem:split-support}, and the other is the unit
  of propositional truncation.
  Every extensional basis is intensional, and so $(\ref{it:IV}) \to (\ref{it:III})$ is clear.
  The implication $(\ref{it:I}) \rightarrow (\ref{it:IV})$ is Lemma~\ref{lem:spec-gives-basis}.
  By Corollary~\ref{prop:basis-img-equiv}, we know that
  $\Equiv{\image{B}{I}}{\CompactOpens{X}}$ for any intensional spectral basis
  $\FamEnum{i}{I}{B_i}$. This implies that any extensional spectral basis is equivalent to the
  basis $\CompactOpens{X} \hookrightarrow X$.
  Conversely, the inclusion $\CompactOpens{X} \hookrightarrow X$ is always a small
  extensional basis, so that we get
  $(\ref{it:IV}) \leftrightarrow (\ref{it:V})$. We also clearly have $(\ref{it:IV}) \rightarrow (\ref{it:IIPrime}) \rightarrow (\ref{it:II})$, which concludes our proof.
\end{proof}

\subsection{Stone locales}\label{subsec:stone}

Clopenness is central to the notion of Stone locale, similar to the fundamental
role played by the notion of a \VCompactOpen{} in the definition of a
\VSpectralLocale{}. To define the clopens, we first define the well-inside
relation.

\begin{definition}[The well-inside relation]\label{defn:well-inside}
  We say that an open $U$ of a locale $X$ is
  \define{well inside} an open $V$, written $U \WellInside V$, if there is an open $W$ with
   $U \wedge W = \Zero{X}$ and $V \vee W = \One{X}$.
\end{definition}

\begin{definition}[Clopen]\label{defn:clopen}
  An open $U$ is called a \define{clopen} if it is \VWellInside{} itself, which
  amounts to saying that it has a Boolean complement. We denote by $\Clopens{X}$ the subtype of $\opens{X}$ consisting of clopens.
\end{definition}

% We denote by $\Clopens{X}$ the type of clopens of $X$:
% \begin{equation*}
%   \Clopens{X} \quad\is\quad \SigmaType{U}{\opens{X}}{\IsClopen{U}}.
% \end{equation*}

Before we proceed to the definition of a Stone locale, we record the following
fact about the \VWellInside{} relation:

\begin{lemma}\label{prop:well-inside-upwards-downwards}
  Given opens $U, V, W : \opens{X}$ of a locale $X$,
  \begin{enumerate}
  \item if $U \WellInside V$ and $V \le W$ then $U \WellInside W$; and
  \item if $U \le V$ and $V \WellInside W$ then $U \WellInside W$.
  \end{enumerate}
\end{lemma}

\begin{definition}[Stone locale]\label{defn:stone-locale}
  A locale $X$ is called \define{Stone} if it satisfies the following conditions:
  \begin{enumerate}[label={$\left({\mathsf{ST{\arabic*}}}\right)$}, leftmargin=4em]
    \item\label{item:stl-compact}%
      It is compact. % a \VCompactLocale{} (i.e.\ the empty meet is compact).
    \item For any $U : \opens{X}$, there exists a small \VDirected{} family
      $\FamEnum{i}{I}{C_i}$, with each $C_i$ \VClopen{}, such that
      $U = \bigvee_{i : I} C_i$.
    \label{item:stl-covering}
  \item The type $\Clopens{X}$ is small.
    \label{item:stl-smallness}
  \end{enumerate}
\end{definition}

\begin{lemma}\label{lem:stone-prop}
  For any locale $X$, being Stone is a proposition.
\end{lemma}
\begin{proof}
  Being compact is a proposition (Lemma~\ref{prop:way-below-prop}) and being
  small is a proposition by Lemma~\ref{lem:being-small-is-prop} (assuming
  univalence). The condition \ref{item:stl-covering} is a proposition since
  the $\PiTySym$-type over a family of propositions is a proposition.
\end{proof}

We denote by $\Stone$ the category of Stone locales with continuous maps as the
morphisms. The defining frame homomorphism of any continuous map of Stone
locales automatically preserves clopens.

The following two lemmas are needed to prove that the \VCompactOpen{}s and the
\VClopen{}s coincide in Stone locales, which we will need later. The proofs are
standard \cite[Lemma~\RSeven{}.3.5]{ptj-ss}. We provide the proof of
Lemma~\ref{prop:way-below-implies-well-inside} for the sake of self-containment,
since it uses our reformulation of the notion of a Stone locale.

\begin{lemma}\label{prop:well-inside-implies-way-below}
  In any \VCompactLocale{}, $U \WellInside V$ implies $U \WayBelow V$ for any
  two opens $U, V$.
\end{lemma}

\begin{lemma}\label{prop:way-below-implies-well-inside}
  In any Stone locale, $U \WayBelow V$ implies $U \WellInside V$ for any
  two opens $U, V$.
\end{lemma}
\begin{proof}
  Let $U, V : \opens{X}$ with $U \WayBelow V$. We know that
  $\IdTy{V}{\bigvee_{i : I} C_i}$ for a family $\FamEnum{i}{I}{C_i}$ consisting of
  \VClopen{}s.
  Since $V \le \bigvee_{i : I} C_i$, it must be the case that there is some
  $k : I$ with $U \le C_k$ as we know $U \WayBelow V$.
  We then have $U \le C_k \WellInside V$ which implies $U \WellInside V$
  by Lemma~\ref{prop:well-inside-upwards-downwards}.
\end{proof}

\begin{corollary}\label{cor:compact-equiv-clopen}
  In any Stone locale, the type $\Equiv{\CompactOpens{X}}{\Clopens{X}}$.
\end{corollary}

\begin{corollary}\label{cor:stone-implies-spectral}
  Every Stone locale $X$ is spectral.%
\end{corollary}

A consequence of Corollary~\ref{cor:compact-equiv-clopen} is that we immediately
get a characterization of Stone locales analogous to the notion of intensional
spectral basis that we gave in Definition~\ref{defn:int-spec-basis}.

\begin{definition}[Basis of clopens]
  An \define{intensional basis of clopens} for a locale $X$ is a small
  intensional basis $\FamEnum{i}{I}{B_i}$ for the frame $\opens{X}$ that
  consists of \VClopen{}s.
  We say that the basis is \emph{extensional} if the map
  $B : I \to \opens{X}$ is
  an \VEmbedding{}. When we say \emph{basis of clopens} without any
  qualification, we mean an extensional basis of clopens.
\end{definition}
Notice that, because clopens are closed under finite joins, the directification
of a basis of clopens is again a basis of clopens. For the next lemma, recall
the notion of extensionalization given in
Definition~\ref{def:extensionalization}.
\begin{lemma}
 If $\FamEnum{i}{I}{B_i}$ is a basis of clopens, then so is its extensionalization.
\end{lemma}

\begin{lemma}\label{lem:basis-to-stone}
  If a locale $X$ is compact and has an unspecified, intensional, and small
  basis of clopens, then it is \VStone{}.
\end{lemma}
\begin{proof}
  Since being Stone is a proposition (by Lemma~\ref{lem:stone-prop}), we can
  work with a specified basis $\FamEnum{i}{I}{B_i}$ with each $B_i$ clopen.
  Condition \ref{item:stl-covering} is immediate since any covering family given
  by the basis is a subfamily of $\Clopens{X}$. It remains to show
  that $\Clopens{X}$ is a small type \ref{item:stl-smallness}. Let
  $C : \Clopens{X}$. Since $X$ is compact, $C$ must be a compact open by
  Lemma~\ref{prop:well-inside-implies-way-below}, and hence must fall in
  the basis by Lemma~\ref{lem:cmp-bsc}. The other direction is direct by construction
  meaning $\Equiv{\Clopens{X}}{\image{B}{I}}$, which concludes that
  $\Clopens{X}$ is small by Remark~\ref{rmk:sr-special-case}.
\end{proof}

\begin{theorem}
  \label{thm:answer-stone}
  The following are logically equivalent for any locale $X$.
  \begin{enumerate}
    \item\label{item:stone-I} $X$ is \VStone.
    \item\label{item:stone-II} $X$ is compact and has an unspecified intensional small basis of clopens.
   \item\label{item:stone-IIPrime} $X$ is compact and has an unspecified extensional small basis of clopens.
      \item\label{item:stone-V} $X$ is compact and the inclusion $\Clopens{X} \hookrightarrow X$ is an extensional
      basis of clopens, where to an open~$U$ we assign the family of all clopens below it.
    \item\label{item:stone-III} $X$ is compact has a specified intensional small basis of clopens.
    \item\label{item:stone-IV} $X$ is compact has a specified extensional small basis of clopens.
  \end{enumerate}
\end{theorem}
Notice that the first four conditions are propositions, but the last two are not in general.
\begin{proof}
  We know by Corollary~\ref{cor:stone-implies-spectral} and
  Corollary~\ref{cor:compact-equiv-clopen} that every Stone
  locale $X$ is spectral and has $\CompactOpens{X} \EquivSym \Clopens{X}$.
  Therefore, by Theorem~\ref{thm:answer}, we know that $\Clopens{X}$ is an
  extensional basis of clopens, which gives the implication
  (\ref{item:stone-I}) $\rightarrow$ (\ref{item:stone-V}).
  The implications (\ref{item:stone-V}) $\rightarrow$ (\ref{item:stone-IV})
  $\rightarrow$ (\ref{item:stone-III}) $\rightarrow$ (\ref{item:stone-II}) are
  direct, and the implication (\ref{item:stone-II}) $\rightarrow$ (\ref{item:stone-I})
  is Lemma~\ref{lem:basis-to-stone}. We also clearly have $(\ref{item:stone-IV}) \rightarrow (\ref{item:stone-IIPrime}) \rightarrow (\ref{item:stone-II})$, which concludes our proof.
\end{proof}

\section{Predicative form of the posetal Adjoint Functor Theorem}
\label{sec:aft}

We start with the definition of the notion of an adjunction in the simplified
context of lattices.

\begin{definition}
  Let $(X, \le_X)$ and $(Y, \le_Y)$ be two preordered sets.
  An adjunction between $X$ and $Y$ consists of a pair of monotonic maps
  $f : X \to Y$ and $g : Y \to X$ satisfying
  \begin{equation*}
    f \dashv g \quad\is\quad \text{$f(x) \le y \leftrightarrow x \le g(y)$ for all $x, y : X$}.
  \end{equation*}
\end{definition}

In locale theory, it is standard convention to denote by $f_* : \opens{X} \to
\opens{Y}$ the right adjoint of a frame homomorphism $f^* : \opens{Y} \to
\opens{X}$ corresponding to a continuous map of locales $f : X \to Y$. The right
adjoint here is known to always exist thanks to a simple application of the
Adjoint Functor Theorem which amounts to the definition:
\begin{equation*}
  f_*(U) \quad\is\quad \bigvee \{ V : \opens{Y} \mid f^*(V) \le U \}.
\end{equation*}
In the predicative setting of type theory, however, it is not clear how the
right adjoint of a frame homomorphism would be defined as the family $\{ V :
\opens{Y} \mid f^*(V) \le U \}$ might be too big in general. This gives rise to the
problem that it is not \emph{a priori} clear that its join in $\opens{X}$
exists. To resolve this problem, we use the assumption of a small basis.%
  The use of a small basis for the posetal Adjoint Functor Theorem in a predicative setting was
  independently observed by Tom de~Jong (personal communication).

\begin{theorem}[Posetal Adjoint Functor Theorem]\label{thm:aft}
  Let $X$ and $Y$ be two locales and
  let $h : \opens{Y} \to \opens{X}$ be a monotone map of frames.
  Assume that $Y$ has some small basis.
  The map $h$ has a right adjoint if and only if it preserves the joins of small
  families.
\end{theorem}
\begin{proof}
  Let $h : \opens{Y} \to \opens{X}$ be a monotone map of frames
  and assume that $Y$ has a small basis $\FamEnum{i}{I}{B_i}$.
  %% Forward.
  The forward direction is easy: suppose $h$ has a right adjoint $g : \opens{X}
  \to \opens{Y}$ and let $\FamEnum{i}{I}{U_i}$ be a family in $\opens{Y}$. By the
  uniqueness of joins, it is sufficient to show that $h(\bigvee_i U_i)$ is the join of
  the family $\FamEnum{i}{I}{h(U_i)}$. It is clearly an upper bound by the fact
  that $h$ is monotone. Given any other upper bound $V$ of
  $\FamEnum{i}{I}{h(U_i)}$,
  we have that $h(\bigvee_i U_i) \le V$ since $h(\bigvee_i U_i) \le V \leftrightarrow \left(\bigvee_i
  U_i\right) \le g(V)$ meaning it is sufficient to show $U_i \le g(V)$ for each $U_i$.
  Since $U_i \le g(V)$ if and only if $h(U_i) \le V$, we are done as the latter can be seen to
  hold directly from the fact that $V$ is an upper bound of the family in
  consideration.

  %% Backward.
  For the converse, suppose $\IdTy{h(\bigvee_i U_i)}{\bigvee_{i : I} h(U_i)}$ for every small
  family of opens $\FamEnum{i}{I}{U_i}$. We define the right adjoint of $h$ as
  \begin{equation*}
    g(U) \quad\is\quad \bigvee \left\{ \beta(b) \mid i : I, h(B_i) \le U \right\}.
  \end{equation*}
  We need to show that $g$ is the right adjoint of $h$ i.e.\ that
  \[h(V) \le U \leftrightarrow V \le g(U)\]
  for any two $V : \opens{Y}$, $U : \opens{X}$.

  For the $(\rightarrow)$ direction, assume $h(V) \le U$. It must be the case that
  $V = \bigvee_{j : J} B_{i_j}$ for
  some specified covering family $\FamEnum{j}{J}{i_j}$.
  This means that we just have to show $B_{i_j} \le g(U)$ for every $j : J$,
  which is the case since $h(B_{i_j}) \le h(V) \le U$.

  For the $(\leftarrow)$ direction, assume $V \le g(U)$. This means that we have:
  \begin{align*}
    h(V) &\quad\le\quad h(g(U))\\
         &\quad=\quad h\left(\bigvee \left\{ B_i \mid i : I, h(B_i) \le U \right\}\right)\\
         &\quad=\quad \bigvee \left\{ h(B_i) \mid i : I, h(B_i) \le U \right\}\\
         &\quad\le\quad U
  \end{align*}
  so that $h(V) \le U$, as required.
\end{proof}

Our primary use case for the Adjoint Functor Theorem is the construction of
Heyting implications in locally small frames with small bases.

\begin{definition}[Heyting implication]\label{defn:heyting-implication}
  Let $X$ be a locale with some small basis. Given any open $U :
  \opens{X}$, the map $(\blank) \wedge U : \opens{X} \to \opens{X}$ preserves joins by
  the frame distributivity law. This means, by Theorem \ref{thm:aft}, that it must
  have a right adjoint $U \Rightarrow (\blank) : \opens{X} \to \opens{X}$. The operation
  $(\blank) \Rightarrow (\blank)$ is known as \emph{Heyting implication}.
\end{definition}

\begin{definition}
  By Theorem~\ref{thm:aft}, for any continuous map $f : X \to Y$ of locales where
  $Y$ has an unspecified small basis, the frame homomorphism $f^* : \opens{Y} \to
  \opens{X}$ has a right adjoint, denoted by
  \begin{equation*}
    f_* : \opens{X} \to \opens{Y}.
  \end{equation*}
\end{definition}

\begin{definition}[Perfect frame homomorphism]\label{defn:perfect-map}
  Let $X$ and $Y$ be two locales and assume that $Y$ has an unspecified small
  basis. A continuous map $f : X \to Y$ is said to be \define{perfect} if $f_*$ is
  \VScottContinuous{}, that is, it preserves small directed suprema.
\end{definition}

Let us also record the following fact about perfect maps that we will need
later.

\begin{lemma}\label{prop:perfect-resp-way-below}
  Let $f : X \to Y$ be a perfect map where $Y$ is a locale with some small basis.
  The frame homomorphism $f^*$ respects the \emph{way below} relation, that is, $U
  \WayBelow V$ implies $f^*(U) \WayBelow f^*(V)$, for any two $U, V : \opens{Y}$.
\end{lemma}

A proof of this fact can be found in \citep{mhe-patch-short} and it works in our
predicative setting.

\begin{corollary}\label{cor:perfect-maps-are-spectral}
  Perfect maps are \VSpectral{} as they preserve the \VCompactOpen{}s.
\end{corollary}

In fact, the converse is also true in the case of spectral locales. The proof
given in~\citep{mhe-patch-short} works, once we know that the required adjoints are available, as established above. We include it in order to illustrate this point.

\begin{lemma}\label{lem:perfect-iff-spectral}
  Let $X$ and $Y$ be locales and assume that $Y$ has some small basis. A
  continuous map $f : X \to Y$ is perfect if and only if it is spectral.
\end{lemma}
\begin{proof}
  The forward direction is given by
  Corollary~\ref{cor:perfect-maps-are-spectral}. For the backward direction,
  assume $f : X \to Y$ to be a spectral map. We have to show that the right adjoint
  $f_* : \opens{X} \to \opens{Y}$ of its defining frame homomorphism is
  \VScottContinuous{}. Letting $\FamEnum{i}{I}{U_i}$ be a directed family in
  $\opens{X}$, we show $f_*(\bigvee_{i : I} U_i) = \bigvee_{i : I} f_*(U_i)$. The
  $\bigvee_{i : I} f_*(U_i) \le f_*(\bigvee_{i : I} U_i)$ direction is easy. For the
  $f_*(\bigvee_{i : I} U_i) \le \bigvee_{i : I} f_*(U_i)$ direction, we appeal
  to Lemma~\ref{prop:spectral-yoneda}. Let $K$ be a compact open
  with $K \le f_*(\bigvee_{i : I} U_i)$. By the adjunction $f^* \dashv f_*$, it must be the
  case that $f^*(K) \le \bigvee_{i : I} U_i$ and since $f^*(K)$ is compact, by the
  spectrality assumption of $f^*$, there must exist some $l : I$ such that $f^*(K)
  \le U_l$. Again by adjointness, $K \le f_*(U_l)$ which implies $K \le \bigvee_{i : I}
  f_*(U_i)$.
\end{proof}

\section{Meet-semilattice of \VScottContinuous{} nuclei}
\label{sec:meet-semilattice}

In this section, we take the first step towards constructing the defining frame
of the patch locale on a spectral locale i.e.\ the frame of \VScottContinuous{}
nuclei. We start by constructing the meet-semilattice of \emph{all} nuclei on a
frame.

\begin{lemma}\label{defn:nuclei-semilattice}
  The type of \VNuclei{} on a given frame $\opens{X}$ forms a meet-semilattice
  under the pointwise order.
\end{lemma}
\begin{proof}
  The top nucleus is defined as the constant map with value $\One{X}$ and
  the meet of two nuclei as $j \wedge k \is U \mapsto j(U) \wedge k(U)$. It is easy to see that $j
  \wedge k$ is the greatest lower bound of $j$ and $k$ so it remains to show that $j \wedge
  k$ satisfies the nucleus laws.

  Inflationarity can be seen to be satisfied from the inflationarity
  of $j$ and $k$ combined with the fact that $j(U) \wedge k(U)$ is the
  greatest lower bound of $j(U)$ and $k(U)$. To see that meet preservation
  holds, let $U, V : \opens{X}$; we have:
  \begin{align*}
    (j \wedge k)(U \wedge V) &\quad=\quad j (U \wedge V) \wedge k (U \wedge V)         \\
                   &\quad=\quad j(U) \wedge j(V) \wedge k(U) \wedge k(V)     \\
                   &\quad=\quad (j(U) \wedge k(U)) \wedge (j(V) \wedge k(V)) \\
                   &\quad=\quad (j \wedge k)(U) \wedge (j \wedge k)(V)
  \end{align*}
  For idempotency, let $U : \opens{X}$. We have:
  \begin{align*}
    (j \wedge k)((j \wedge k)(U)) &\quad=\quad j (j(U) \wedge k(U)) \wedge k(j(U) \wedge k(U))      \\
                        &\quad=\quad j(j(U)) \wedge j(k(U)) \wedge k(j(U)) \wedge k(k(U)) \\
                        &\quad\le\quad j(j(U)) \wedge k(k(U))                     \\
                        &\quad=\quad j(U) \wedge k(U)                           \\
                        &\quad=\quad (j \wedge k)(U)
  \end{align*}
\end{proof}

We now show that this meet-semilattice can be \emph{refined} to the \VSCNuclei{}
(i.e.\ the \emph{perfect} nuclei).

\begin{lemma}\label{prop:sc-nuclei-semilattice}
  The \VSCNuclei{} on any locale form a meet-semilattice.
\end{lemma}
\begin{proof}
  Let $X$ be a locale. The construction is the same as the one from
  Lemma~\ref{defn:nuclei-semilattice}; the top element is constant map with
  value $\One{X}$, which is trivially \VScottContinuous{} so it remains to show
  that the meet of two Scott continuous nuclei is \VScottContinuous{}. Consider
  two Scott continuous nuclei $j$ and $k$ on $\opens{X}$ and a directed small
  family $\FamEnum{i}{I}{U_i}$. We then have:
  \begin{align*}
       (j \wedge k) \paren{\bigvee_{i : I} U_i}
  &\quad\equiv\quad j \paren{\bigvee_{i : I} U_i} \wedge k \paren{\bigvee_{j : I} U_j}
     && \\
  &\quad=\quad \paren{\bigvee_{i : I} j(U_i)} \wedge \paren{\bigvee_{j : I} k(U_j)}
     && [\text{Scott continuity of $j$ and $k$}]\\
  &\quad=\quad \bigvee_{(i, j) : I \times I} j(U_i) \wedge k(U_j)
     && [\text{distributivity}]\\
  &\quad=\quad \bigvee_{i : I} j(U_i) \wedge k(U_i)
     && [\text{\dag}] \\
  &\quad\equiv\quad \bigvee_{i : I} (j \wedge k)(U_i)
     && [\text{meet preservation}]
  \end{align*}
  where, for the (\dag)\ step, we use antisymmetry. The $(\ge)$ direction is immediate.
  For the $(\le)$ direction, we need to show that $\bigvee_{(i, j) : I \times I} j(U_i) \wedge
  k(U_j) \le \bigvee_{i : I} j(U_i) \wedge k(U_i)$, for which it suffices to show that $\bigvee_{i :
  I} j(U_i) \wedge k(U_i)$ is an upper bound of $\{j(U_i) \wedge k(U_j)\}_{(i, j) : I \times I}$.
  Let $m, n : I$ be two indices. As $\FamEnum{i}{I}{U_i}$ is directed, there must
  exist some $o$ such that $U_o$ is an upper bound of $\{U_m, U_n\}$. Using the
  monotonicity of $j$ and $k$, we get $j(U_m) \wedge k(U_n) \le j(U_o) \wedge k(U_o) \le \bigvee_{i :
  I} j(U_i) \wedge k(U_i)$ as desired.
\end{proof}

\section{Joins in the frame of \VScottContinuous{} nuclei}
\label{sec:joins}

The nontrivial component of the patch frame construction is the join of a family
$\FamEnum{i}{I}{k_i}$ of \VSCNuclei{}, as the pointwise join fails to be
idempotent in general, and not inflationary when the family in consideration is
empty. When the family in consideration is directed, however, the pointwise join
is idempotent, and so it is the join in the poset of Scott continuous nuclei.
\begin{lemma}\label{lem:directed-computed-pointwise}
  Given a directed family $\FamEnum{i}{I}{k_i}$ of \VSCNuclei{}, their join
  is computed pointwise, that is, $\left(\bigvee_{i : I} k_i\right)(U) = \bigvee_{i : I}
  k_i(U)$.
\end{lemma}
\begin{proof}
  The argument given in the paragraph preceding
  \citep[Lemma 3.1.8]{mhe-properly-injective} works in our setting.
\end{proof}
Regarding arbitrary joins, the situation is more complicated. A construction of
the join, due to \cite{mhe-properly-injective}, is based on the idea that, if we
start with a family $\FamEnum{i}{I}{k_i}$ of nuclei, we can consider the family
with index type $\ListTy{I}$ of words over $I$, defined by
\begin{equation*}
  (i_0i_1 \cdots i_{n-1}) \mapsto k_{i_{n-1}} \circ \cdots \circ k_1 \circ k_0.
\end{equation*}
This family is easily seen to be direced.

To talk about such families of finite compositions over a given family of
(pre)nuclei, we adopt the following notation.
\begin{itemize}
  \item $\emptyl$ denotes the empty word.
  \item Given words $s : \ListTy{I}$ and $t : \ListTy{I}$, the concatenation
    is denoted by $s \append t$.
\end{itemize}

To define the join operation, we will use the iterated composition function
$\ddnm$ that we define as follows:

\begin{definition}[Family of finite compositions]
  Given a small family $\FamEnum{i}{I}{k_i}$ of nuclei on a given locale $X$,
  its \define{family of finite compositions} is the family defined by
  \begin{equation*}
    k^*(i_0\cdots i_{n-1}) \quad\is\quad k_{i_{n-1}} \circ \cdots \circ k_{i_0}.
  \end{equation*}
  By an abuse of notation, we omit the superscript ``$*$'' when there is no possibility of confusion.
\end{definition}
The order of composition is actually not important.
A finite composition $k_s$ is, in general, not a nucleus. It is, however, always
a prenucleus which we prove below.

\begin{lemma}\label{lem:star-prenucleus}
  Given a family $\FamEnum{i}{I}{k_i}$ of nuclei on a locale,
  $k_{s}$ is a prenucleus for every $s : \ListTy{I}$.
\end{lemma}
\begin{proof}
  If $s = \emptyl$, we are done as it is immediate that the identity function
  $\mathsf{id}$ is a prenucleus. If $s = i \cons s'$, we need to show that
  $\dd{s'} \circ k_i$ is a prenucleus. For meet preservation, let $U, V
  : \opens{X}$.
  \begin{align*}
    (\dd{s'} \circ k_i)(U \wedge V)
    &\quad=\quad \dd{s'}(k_i(U \wedge V))                     \\
    &\quad=\quad \dd{s'}(k_i(U) \wedge k_i(V))
    && [\text{$k_i$ is a nucleus}]          \\
    &\quad=\quad \dd{s'}(k_i(U)) \wedge \dd{s'}(k_i(V))
    && [\text{inductive hypothesis}]        \\
    &\quad=\quad (\dd{s'} \circ k_i)(U) \wedge (\dd{s'} \circ k_i)(V)
  \end{align*}
  For inflationarity, consider some $U : \opens{X}$. We have that $U \le
  k_i(U) \le \dd{s'}(k_i(U))$, by the inflationarity property of $k_i$ and the
  inductive hypothesis.
\end{proof}

\begin{lemma}\label{prop:star-ub}
  Given a nucleus $j$ and a family $\FamEnum{i}{I}{k_i}$ of \VNuclei{} on a
  locale, if $j$ is an upper bound then it is also an upper bound of the
  family of finite compositions.
\end{lemma}
\begin{proof}
  Let $j$ and $\FamEnum{i}{I}{k_i}$ be, respectively, a \VNucleus{} and a
  family of \VNuclei{} on a locale. Let $s : \mathsf{List}(I)$. We proceed by
  induction on $s$. If $s = \emptyl$, we have that $\mathsf{id}(U) \equiv U \le j(U)$. If
  $s = i \cons s'$, we then have:
  \begin{align*}
    k_{s'}(k_i(U))
    \quad&\le\quad k_{s'}(j(U))
    && [\text{monotonicity of $k_{s'}$ (Lemma~\ref{lem:star-prenucleus} and monotonicity of prenuclei)}] \\
    \quad&\le\quad j(j(U))
    && [\text{IH}] \\
    \quad&\le\quad j(U)
    && [\text{idempotency of $j$}]
  \end{align*}
\end{proof}

%% Called `^*-scott-continuous` in the formalisation.
\begin{lemma}\label{prop:star-sc}
  Given a family $\FamEnum{i}{I}{k_i}$ of \VSCNuclei{} on a locale,
  the prenucleus $k_s$ is \VScottContinuous{}, for every $s : \ListTy{I}$
\end{lemma}
\begin{proof}
  Any composition of finitely many \VScottContinuous{} functions is
  \VScottContinuous{}.
\end{proof}

\begin{lemma}\label{lem:star-dir}
  Given a family $\FamEnum{i}{I}{k_i}$ of \VNuclei{} on a locale, the family of
  finite compositions is \VDirected{}.
\end{lemma}
\begin{proof}
  The family $\FamEnum{s}{\ListTy{I}}{k_{s}}$ is indeed always inhabited by
  the identity nucleus given by $\dd{\epsilon}$. The upper bound of nuclei $\dd{s}$ and
  $\dd{t}$ is given by $\dd{s \append t}$, which is equal to $\dd{t} \circ \dd{s}$.
  The fact that this is an upper bound of $\{ \dd{s}, \dd{t} \}$ follows from
  monotonicity and inflationarity.
\end{proof}

\begin{lemma}\label{lem:delta-gamma}
  Let $j$ be a nucleus and $\FamEnum{s}{\ListTy{I}}{k_{s}}$ a family of
  \VNuclei{} on a locale $X$. Consider the family of finite compositions over
  the family $\FamEnum{i}{I}{j \wedge k_i}$. Each finite composition
  $(j \wedge k)^*_{s}$ is a lower bound of $\{ k_{s}, j \}$ for every
  $s : \ListTy{I}$.
\end{lemma}

We are now ready to construct the join operation in the meet-semilattice of
\VSCNuclei{} hence defining the patch frame $\opens{\Patch(X)}$ of the frame of
a locale $X$.

\begin{theorem}[Join of \VSCNuclei{}]\label{defn:sc-join}
  Let $\FamEnum{i}{I}{k_i}$ be a family of \VSCNuclei{}. The join of $K$ can be
  calculated as
  \begin{equation*}
    \left(\bigvee_{i : I} k_i\right) (U) \is \bigvee_{s : \ListTy{I}} k_{s}(U).
  \end{equation*}
\end{theorem}
\begin{proof}
  It must be checked that this is (1) indeed the join, (2) is a \VSCNucleus{}
  i.e.\ it is inflationary, binary-meet-preserving, idempotent, and
  \VScottContinuous. The inflationarity property is direct. For meet preservation,
  consider some $U, V : \opens{X}$. We have:
  \begin{align*}
        \left(\bigvee_{i : I} k_i\right)(U \wedge V)
    &\quad=\quad \bigvee_{s : \ListTy{I}} \dd{s}(U \wedge V)   \\
    &\quad=\quad \bigvee_{s : \ListTy{I}} \dd{s}(U) \wedge \dd{s}(V)
        && [\text{Lemma~\ref{prop:star-sc}}]  \\
    &\quad=\quad \bigvee_{s,t : \ListTy{I}} \dd{s}(U) \wedge \dd{t}(V)
        && [\dag]                           \\
    &\quad=\quad \paren{\bigvee_{s : \ListTy{I}} \dd{s}(U)} \wedge \paren{\bigvee_{t : \ListTy{I}} \dd{t}(V)}
        && [\text{distributivity}]       \\
    &\quad=\quad \paren{\bigvee_{i : I} k_i}(U) \wedge \paren{\bigvee_{i : I} k_i}(V)
  \end{align*}
  where step ($\dag$) uses antisymmetry. The
  $(\le)$ direction is direct whereas for the $(\ge)$
  direction, we show that $\bigvee_{s : \ListTy{I}} \dd{s}(U) \wedge \dd{s}(V)$ is an upper
  bound of the family
  %\begin{equation*}
    $\FamEnum{s,t}{\ListTy{I}}{\dd{s}(U) \wedge \dd{t}(V)}.$
  %\end{equation*}
  Consider arbitrary $s, t : \ListTy{I}$. By the directedness of the family of
  finite compositions we know that there exists some $u : \ListTy{I}$ such that
  $\dd{u}$ is an upper bound of $\{\dd{s}, \dd{t}\}$.
  We then have
  \begin{equation*}
    \dd{s}(U) \wedge \dd{t}(V) \le \dd{u}(U) \wedge \dd{u}(V) \le \bigvee_{s : \ListTy{I}} \dd{s}(U) \wedge \dd{s}(V).
  \end{equation*}
  For idempotency, let $U : \opens{X}$. We have that:
  \begin{align*}
         \left(\bigvee_{i} k_i\right)\left(\left(\bigvee_{i} k_i\right)(U)\right)
    &\quad\equiv\quad \bigvee_{s : \ListTy{I}} \dd{s}\left( \bigvee_{t : \ListTy{I}} \dd{t}(U) \right) \\
    &\quad=\quad \bigvee_{s : \ListTy{I}} \bigvee_{t : \ListTy{I}} \dd{s}(\dd{t}(U)) & [\text{Lemma~\ref{prop:star-sc}}]\\
    &\quad\le\quad \bigvee_{s,t : \ListTy{I}} \dd{s}(\dd{t}(U)) & ~ \\
    &\quad\le\quad \bigvee_{s : \ListTy{I}} \dd{s}(U) & [\dag]\\
    &\quad\equiv\quad \paren{\bigvee_i k_i}(U),
  \end{align*}
  where for the step ($\dag$) it suffices to show that
  $\bigvee_{s : \ListTy{I}} \dd{s}(U)$ is an upper bound
  of the family
  $\FamEnum{s,t}{\ListTy{I}}{\dd{s}(\dd{t}(U))}.$
  %\[\setof{ \dd{s}(\dd{t}(U)) \mid \Pair{s}{t} : \ListTy{I} \times \ListTy{I} }.\]
  %
  The prenucleus $\dd{s \append t}$ is an upper bound of $\dd{s}$ and $\dd{t}$
  (as in Lemma~\ref{lem:star-dir}).
  We have that
  \[\dd{s}(\dd{t}(U)) = \dd{s \append t}(U) \le \bigvee_{s : \ListTy{I}}\dd{s}(U).\]

  For Scott continuity, let $\FamEnum{j}{J}{U_j}$ be a directed family over
  $\opens{X}$. We then have:
  \begin{align*}
    \left(\bigvee_{i : I} k_i\right)\left(\bigvee_{j : J} U_j\right)
  &\quad\equiv\quad \bigvee_{s : \ListTy{I}} \dd{s}\left(\bigvee_{j : J} U_j\right) \\
  &\quad=\quad \bigvee_{s : \ListTy{I}} \bigvee_{j : J} \dd{s}(U_j) && [\text{Lemma~\ref{prop:star-sc}}]    \\
  &\quad=\quad \bigvee_{j : J} \bigvee_{s : \ListTy{I}} \dd{s}(U_j) && [\text{joins commute}] \\
  &\quad\equiv\quad \bigvee_{j : J} \left(\bigvee_{i : I} k_i\right)(U_j)
  \end{align*}
  as required.

  The fact that $\bigvee_i k_i$ is an upper bound of $\FamEnum{i}{I}{k_i}$ is easy
  to verify. To see that it is \emph{the least} upper bound, consider a
  \VSCNucleus{} $j$ that is an upper bound of $\FamEnum{i}{I}{k_i}$. Let $U :
  \opens{X}$. We need to show that $(\bigvee_i k_i)(U) \le j(U)$. We know by
  Lemma~\ref{prop:star-ub} that $j$ is an upper bound of the family of
  finite compositions, since it is
  an upper bound of $\FamEnum{i}{I}{k_i}$, which is to say
  $k_{s}(U) \le j(U)$ for every $s : \ListTy{I}$ i.e.\ $j(U)$ is an upper bound
  of the family $\FamEnum{s}{\ListTy{I}}{k_s(U)}$.
  Since $\left(\bigvee_i k_i\right)(U)$ is the least upper bound of this family by
  definition, it follows that it is below $j(U)$.
\end{proof}

We use Lemma~\ref{lem:delta-gamma} to prove the distributivity law.

\begin{lemma}[Distributivity]\label{prop:distributivity}
  For any \VSCNucleus{} $j$ and any family $\FamEnum{i}{I}{k_i}$ of
  \VSCNuclei{}, we have the equality:
  \begin{equation*}
    j \wedge \left(\bigvee_{i : I} k_i\right) = \bigvee_{i : I} j \wedge k_i.
  \end{equation*}
\end{lemma}

It follows that the \VScottContinuous{} \VNuclei{} form a frame.

\begin{definition}[Patch locale of a spectral locale]\label{defn:patch}
  The \emph{patch locale} of a spectral locale $X$%
  , written $\Patch(X)$,
  is given by the frame of \VSCNuclei{} on $X$.
\end{definition}

Notice that the truth value of the relation $j \le k$ between Scott continuous nuclei lives by default in the universe $\USucc{\UU}$. However, it has an equivalent copy in the universe $\UU$:

\begin{lemma}\label{lem:nuclei-basic-order} \label{prop:basic-ordering-iff}
  For any spectral locale $X$ and any two nuclei $j , k : \Patch(X)$, we have that
  \begin{equation*}
    \text{$j \le k$ if and only if $j(K) \le k(K)$ for all $K:\CompactOpens{X}$},
  \end{equation*}
  and hence $\Patch(X)$ is locally small.
\end{lemma}
\begin{proof}
  The usual pointwise ordering obviously implies the basic ordering so we
  address the other direction. Let $j$ and $k$ be two \VSCNuclei{} on
  a \VSpectralLocale{} $X$ and assume $j(K) \le k(K)$, for all $K : \CompactOpens{X}$.
  We need to show that $j(U)
  \le k(U)$ for every open $U$ so let $U : \opens{X}$.
  By \ref{item:sl-covering}, $U$ can be decomposed as
  $U = \bigvee_{i : I} K_i$
  for some directed covering family $\FamEnum{i}{I}{K_i}$ consisting of
  \VCompactOpen{}s.
  We then have
  $j(\bigvee_{i : I} K_i) = \bigvee_{i : I} j(K_i)$ by Scott continuity and
  \[\bigvee_{i : I} j(K_i) \le \bigvee_{i : I} k(K_i)\]
  since $j(K_i) \le k(K_i)$ for every $i : I$.
  Finally, because the quantification is over the small type $\CompactOpens{X}$, the proposition
  ``$\text{$j(K) \le k(K)$ for all $K:\CompactOpens{X}$}$''
  is small.
\end{proof}

\section{The coreflection property of \PatchTEXorPDF{}}
\label{sec:coreflection}

We prove in this section that our construction of $\Patch$ has the desired
universal property: it exhibits $\Stone$ as a coreflective subcategory of
$\Spec$.
The notions of \emph{closed} and \emph{open} \VNuclei{} are crucial for proving
the universal property. We first give the definitions of these. Let $U$ be an
open of a locale $X$;
\begin{enumerate}
  \item The \define{closed nucleus} induced by $U$ is the map $V \mapsto U \vee V$;
  \item The \define{open nucleus} induced by $U$ is the map $V \mapsto U \Rightarrow V$.
\end{enumerate}

We denote the closed nucleus induced by open $U$ by $\closednucl{U}$ and the
open nucleus induced by $U$ by $\opennucl{U}$.
\begin{lemma}\label{lem:characterisation}
  For any spectral locale $X$ and any monotone map $h : \opens{X} \to \opens{X}$,
  if for every $U : \opens{X}$ and compact $K : \opens{X}$ with $K \le h(U)$,
  there is some compact $K' \le U$ such that $K \le h(K')$, then $h$ is
  \VScottContinuous{}
\end{lemma}
\begin{proof}
It suffices to show the relation $h( \bigvee_{i : I} U_i ) \leq \bigvee_{ i : I} h(U_i)$ holds. Since $X$ is spectral, let $(K_j)_{j : J}$ be a small family of compact opens such that $\bigvee_{j : J} K_j = h( \bigvee_{i : I} U_i )$. For any $j : J$, we have $K_j \leq h( \bigvee_{i : I} U_i )$, so by assumption there is a compact open $K\leq  \bigvee_{i : I} U_i$ such that $K_j \leq h(K)$. By compactness of $K$ there is an $i : I$ such that $K\leq U_i$, and so $K_j\leq h(K)\leq h(U_i)\leq \bigvee_{i:I} h(U_i)$ and so we can take $K' \is K_j$.
\end{proof}

\begin{lemma}
  For any open of any locale, the closed nucleus $\closednucl{U}$ is \VScottContinuous{},
  whereas the open nucleus $\opennucl{U}$ is \VScottContinuous{} if
  the open $U$ is \VCompact{}.
\end{lemma}
\begin{proof}
  The Scott continuity of the closed nucleus is easy to see.
  For the open nucleus, let $D$ be a \VCompact{} open of a locale. By
  Lemma~\ref{lem:characterisation}, it is sufficient to show that for any open $V$
  and any \VCompact{} open $C_1$ with $C_1 \le D \Rightarrow V$, there exists some compact
  $C_2 \le D$ such that $C_1 \le D \Rightarrow C_2$. Let $V$ and $C_1$ be two opens with $C_1$
  compact and satisfying $C_1 \le D \Rightarrow V$. Pick $C_2 \is D \wedge C_1$. We know that this
  is compact by \ref{item:sl-coherent}. It remains to check (1) $C_2 \le V$ and (2)
  $C_1 \le D \Rightarrow C_2$, both of which are direct.
\end{proof}

In Lemma~\ref{lem:eps-perfect}, we prove that the map whose inverse image sends
an open $U$ to the closed nucleus $\closednucl{U}$ is perfect. Before
Lemma~\ref{lem:eps-perfect}, we record an auxiliary result:

\begin{lemma}\label{lem:eps-ra-bot}
  Let $X$ be a spectral locale. The right adjoint $\varepsilon_* : \opens{\Patch(X)} \to
  \opens{X}$ to the closed-nucleus formation operation $\closednucl{\blank}$ is
  given by $\varepsilon_*(j) = j(\ZeroSym)$ for every \VSCNucleus{} $j$ on~$X$.
\end{lemma}

%\begin{lemma}\label{lem:directed-computed-%pointwise}
 % Given a directed family $\FamEnum{i}{I}%{k_i}$ of \VSCNuclei{}, their join
%  is computed pointwise, that is, $\left(\bigvee_{i : I} k_i\right)(U) = \bigvee_{i : I}
%  k_i(U)$.
%\end{lemma}

The proof of Lemma~\ref{lem:eps-ra-bot} can be found in
\cite{mhe-patch-short}. It is omitted here as it is  unchanged in our
predicative setting.

\begin{lemma}\label{lem:eps-perfect}
  The function $\closednucl{\blank}$ is a perfect frame homomorphism $\opens{X}
  \to \opens{\Patch(X)}$.
\end{lemma}
\begin{proof}
  We have to show that the right adjoint $\varepsilon_*$ of $\closednucl{\blank}$ is
  \VScottContinuous{}. Let $\FamEnum{i}{I}{k_i}$ be a directed family of
  \VSCNuclei{}. Thanks to Lemma~\ref{lem:eps-ra-bot}, it suffices to show
  \[\left(\bigvee_{i : I} k_i\right)(\ZeroSym) = \bigvee_{i : I} \varepsilon_*(k_i).\] By
  Lemma~\ref{lem:directed-computed-pointwise}, we have
  that $\left(\bigvee_{i : I} k_i\right)(\ZeroSym) = \bigvee_{i : I} k_i(\ZeroSym)$.
\end{proof}

This function defines a continuous map $\varepsilon : \Patch(X) \to X$, which we will
show to be the counit of the coreflection in consideration.

\subsection{\PatchTEXorPDF{} is Stone}

Before we proceed to showing that the patch locale has the desired universal
property, we first need to show that $\Patch(X)$ is Stone for any spectral
locale $X$. For this purpose, we need to (1) show that it is compact, and (2)
construct a basis for it consisting of clopens. We start with compactness.

\begin{lemma}\label{prop:patch-is-compact}
  $\Patch(X)$ is a \VCompactLocale{} for any spectral locale $X$.
\end{lemma}
\begin{proof}
  Recall that the top element
  $\OneSym$ of $\Patch(X)$ is defined as $\One{\Patch} \is U \mapsto \One{X}$.
  Because $\varepsilon^*$ is a frame
  homomorphism, it must be the case that $\One{\Patch} = \varepsilon^*(\One{X})$
  meaning it suffices to show $\varepsilon^*(\One{X}) \WayBelow \varepsilon^*(\One{X})$. By
  Lemma~\ref{prop:perfect-resp-way-below}, it suffices to show
  $\One{X} \WayBelow \One{X}$
  which is immediate by \ref{item:sl-compact}.
\end{proof}

To construct a basis consisting of clopens, we will use the following fact,
which was already mentioned above:

\begin{lemma}\label{prop:complementation}
  The open nucleus $\opennucl{U}$ is the Boolean complement of the closed
  nucleus $\closednucl{U}$.
\end{lemma}

\begin{lemma}\label{lem:johnstones-lemma}\label{lem:last-step}
  Let $X$ be a spectral locale.
  Given any perfect nucleus $j : \Patch(X)$, we have that
  \begin{equation*}
    j = \bigvee_{K: \CompactOpens{X}} \closednucl{j(K)} \wedge \opennucl{K}  =
    \bigvee_{\stackrel{K_1,K_2: \CompactOpens{X}}{K_1 \le j(K_2)}} \closednucl{j(K_1)} \wedge \opennucl{K_2}.
    % =
    % \bigvee \left\{
    %    \closednucl{K_1} \wedge \opennucl{K_2} \mid K_1,K_2 : \CompactOpens{X},\ K_1 \le j(K_2)
    %  \right\}.
  \end{equation*}
\end{lemma}
\begin{proof}
   The second equality in the statement is clear, so let us show the first one.
   We use the fact, proved in \cite[Lemma II.2.7]{ptj-ss}, that for any nucleus $j$ on any locale,
 \begin{equation*}
    j = \bigvee_{U: \opens{X}} \closednucl{j(U)} \wedge \opennucl{U}.
  \end{equation*}
  Suppose now additionally that $X$ is spectral and the nucleus $j$ is  Scott continuous. The inequality $\bigvee_{K: \CompactOpens{X}} \closednucl{j(K)} \wedge \opennucl{K} \leq \bigvee_{U: \opens{X}} \closednucl{j(U)} \wedge \opennucl{U} = j$ is trivial, so let us show the reverse one.
Let $K : \CompactOpens{X}$ and notice that $ \left(\closednucl{j(K)} \wedge \opennucl{K}\right)(K)= (j(K)\vee K)\wedge (K\Rightarrow K)= j(K)$. Therefore
$$j(K)=  \left(\closednucl{j(K)} \wedge \opennucl{K}\right)(K) \leq \left( \bigvee_{K': \CompactOpens{X}} \closednucl{j(K')} \wedge \opennucl{K'} \right) (K),$$
so the required relation follows from Lemma~\ref{lem:nuclei-basic-order}.
\end{proof}

%\begin{lemma}\label{lem:last-step}
%  Let $X$ be a spectral locale.
%  Given any perfect nucleus $j : \Patch(X)$, % we have that
%  \begin{equation*}
%    \bigvee_{K : \CompactOpens{X}} % \closednucl{j(K)} \wedge \opennucl{K}
%    =
 %   \bigvee \left\{
  %      \closednucl{K_1} \wedge \opennucl{K_2} \mid K_1,K_2 : \CompactOpens{X},\ K_1 \le j(K_2)
 %     \right\}.
 % \end{equation*}
% \end{lemma}

\begin{theorem}\label{thm:patch-is-stone}
  $\Patch(X)$ is a Stone locale, for every spectral locale $X$.
\end{theorem}
\begin{proof}
  Compactness \ref{item:stl-compact} was given in Lemma~\ref{prop:patch-is-compact}.
  For \ref{item:stl-covering}, let $j : \Patch(X)$. We take the intensional basis
  \begin{align*}
    B \quad&\is\quad \left( \closednucl{K_1} \wedge \opennucl{K_2} \right)_{(K_1, K_2) : (\CompactOpens{X} \times \CompactOpens{X})}\\
    I_j \quad&\is\quad \left(\left(\SigmaType{K_1, K_2}{\CompactOpens{X}}{K_1 \le j(K_2)}\right), \gamma \right)\ \text{where},\\
    \gamma(K_1, K_2)\quad&\is\quad \closednucl{K_1} \wedge \opennucl{K_2}
  \end{align*}
  Even though this family is not a priori directed, we know that it can be
  directified as explained in Remark~\ref{rmk:directification}. We therefore
  obtain a specified, small, intensional basis of clopens by
  Lemma~\ref{lem:last-step}. We denote the
  directed form of the basis by $B^{\uparrow}$ which is small by
  Remark~\ref{rmk:directification} since $\CompactOpens{X}$ is small by
  \ref{item:sl-smallness}. Since in a compact locale, clopens are compact by
  Lemma~\ref{prop:well-inside-implies-way-below}, and so the clopens fall in the
  basis by Lemma~\ref{lem:cmp-bsc}. Therefore, we get the equivalence of types
  \[\Clopens{X} \EquivSym \image{B^{\uparrow}}{_},\]
  which concludes \ref{item:stl-smallness} by Lemma~\ref{prop:img-small}.
\end{proof}

\subsection{The universal property of \PatchTEXorPDF{}}
\label{subsec:universal}

We now show that $\Patch$ is the right adjoint to the inclusion $\Stone \hookrightarrow \Spec$.

\begin{lemma}\label{lem:adjoint-to-induced-map}
    Given any spectral map $f : X \to A$ from a Stone locale into a
  \VSpectralLocale{},
   define a map $\overline{f}^* : \opens{\Patch(A)} \to \opens{X}$ by
  \begin{equation*}
    \bar{f}^*(j) \quad\is\quad \bigvee_{K : \CompactOpens{A}} f^*(j(K)) \wedge \neg f^*(K).
  \end{equation*}
  Then, the map $\bar{f}_* : \opens{X} \to \opens{\Patch(A)}$ defined by
  \begin{equation*}
    \bar{f}_* (V) \quad\is\quad f_* \circ \closednucl{V} \circ f^*
  \end{equation*}
  is the right adjoint of $\bar{f}^*$.
\end{lemma}
\begin{proof}
 First, note that $f_* \circ \closednucl{V} \circ f^*$ is indeed a \VSCNucleus{},
  and both $\bar{f}^*$ and $\bar{f}_*$ are clearly monotone. Let us first
  prove the forward implication.
  Assume that for $j : \opens{\Patch(A)}$ and $V : \opens{X}$
  the relation $\bar f^*(j)\leq V$ holds.
  By Lemma~\ref{lem:nuclei-basic-order}, in order to show that $j \leq
  \bar{f}_*(V)$, it suffices to show that for any compact open $K :
  \CompactOpens{A}$, the inequality $j(K) \leq \bar{f}_*(V)(K)$ holds. Hence, let $K
  : \CompactOpens{A}$, and note that
  \begin{equation*}
    f^*(j(K))\wedge \neg f^*(K) \leq \bar f^*(j)\leq V
  \end{equation*}
  Notice that $f^*(K)$ is clopen, by the fact that $f^*$ is a spectral map and
  Lemma~\ref{prop:way-below-implies-well-inside}.
  It is therefore complemented in the lattice $\opens{X}$, and so
  we have $f^*(j(K))\leq V\vee f^*(K) = \closednucl{V}(f^*(K))$, which by
  adjunction yields $j(K)\leq f_*(\closednucl{V}(f^*(K)))= \bar f_{*}(K)$, as
  required.

  Let us now show the reverse implication.
  Let $ j : \opens{\Patch(A)}$ and $V : \opens{X}$
  and assume that $j\leq \bar f_*(V)$.
  Once again, by the definition of the ordering on $\opens{\Patch(A)}$, for all
  $K : \CompactOpens{A}$ we have $j(K) \leq f_*(V \vee f^*(K))$, which by adjunction
  equivalently yields $f^*(j(K))\leq V \vee f^*(K)$. Since $f^*(K)$ is clopen, and
  hence complemented
  in the lattice $\opens{X}$ it follows that $f^*(j(K)) \wedge \neg f^*(K) \leq V$. Hence,
  $\bar f^*(j)\leq V$.
\end{proof}

\begin{theorem}\label{thm:main}
  Given any spectral map $f : X \to A$ from a Stone locale into a
  \VSpectralLocale{}, there exists a unique spectral map $\bar{f} : X \to \Patch(A)$
  satisfying $\varepsilon \circ \bar{f} = f$, as illustrated in the following diagram in
  $\Spec$:
  \begin{center}
    \begin{tikzcd}
      X \arrow[d, swap, "f"] \arrow[dr, dashed, "\bar{f}"] & \\
      A & \Patch(A) \arrow[l, "\varepsilon"]
    \end{tikzcd}
  \end{center}
\end{theorem}
\begin{proof}
Assume that a locale map $\bar f : X \to \Patch(A)$ satisfies the condition in the theorem. Then, for any $j : \opens{\Patch(A)}$ one has
 \begin{align*}
   \bar f^*(j)
  &\quad=\quad \bar f^* \left( \bigvee_{K : \CompactOpens{A}}\closednucl{j(K)} \wedge \opennucl{K}\right) && [\text{Lemma~\ref{lem:johnstones-lemma}}]\\
  &\quad=\quad  \bigvee_{K : \CompactOpens{A}} \bar f^*\left( \closednucl{j(K)} \wedge \opennucl{K}\right) && [\text{$\bar f^*$ preserves small joins}]    \\
  &\quad=\quad  \bigvee_{K : \CompactOpens{A}} \bar f^*\left( \closednucl{j(K)}\right) \wedge \neg \bar f^*\left(\closednucl{K}\right) && [\text{$\bar f^*$ preserves binary meets and complements}] \\
  & \quad=\quad  \bigvee_{K : \CompactOpens{A}}  f^*\left( j(K)\right) \wedge \neg  f^*\left( K\right) && [\text{commutativity of the diagram}]
  \end{align*}
and hence $\bar f$ is uniquely determined. If we now   define a monotone map $\overline{f}^* : \opens{\Patch(A)} \to \opens{X}$ by
$$
    \bar{f}^*(j) \quad\is\quad \bigvee_{K : \CompactOpens{A}} f^*(j(K)) \wedge \neg f^*(K),
$$ it is easy to show it preserves the top element (namely the top nucleus with constant value $\One{A}$) because $\ZeroSym_A$ is compact. It also preserves binary (pointwise) meets as
 \begin{align*}
   & \bar f^*(j_1)\wedge \bar f^*(j_2)\\
   &=  \bigvee_{K_1,K_2 : \CompactOpens{A}}  f^*\left( j_1(K_1)\right) \wedge \neg  f^*\left( K_1\right) \wedge f^*\left( j_2(K_2)\right) \wedge \neg  f^*\left( K_2\right) && [\text{distributivity}]   \\
  &=   \bigvee_{K_1,K_2 : \CompactOpens{A}}  f^*\left( j_1(K_1) \wedge j_2(K_2)\right) \wedge \neg  f^*\left( K_1\right)  \wedge \neg  f^*\left( K_2\right) && [\text{$f^*$ preserves binary meets}]    \\
  &\leq   \bigvee_{K_1,K_2 : \CompactOpens{A}}  f^*\left( j_1(K_1\vee K_2) \wedge j_2(K_1\vee K_2)\right) \wedge \neg  f^*\left( K_1\right)  \wedge \neg  f^*\left( K_2\right)  && [\text{monotonicity}] \\
  &= \bigvee_{K_1,K_2 : \CompactOpens{A}}  f^*\left( (j_1\wedge j_2)(K_1\vee K_2)\right) \wedge \neg  f^*\left( K_1\right)  \wedge \neg  f^*\left( K_2\right) && ~ \\
  &= \bigvee_{K_1,K_2 : \CompactOpens{A}}  f^*\left( (j_1\wedge j_2)(K_1\vee K_2)\right) \wedge \neg  f^*\left( K_1 \vee K_2\right)  && [\text{De Morgan law}]  \\
  &= \bar f^*(j_1\wedge j_2). && [\text{$\CompactOpens{X}$ closed under $(\blank) \wedge (\blank)$}]
  \end{align*}

 Moreover, Lemma~\ref{lem:adjoint-to-induced-map} and Theorem~\ref{thm:aft} ensure that $\bar f^*$ preserves small joins and so it is a frame homomorphism.

Let us finally show that $\bar f$ makes the diagram commute. Since compact opens form a small basis of $A$, it suffices to show that $\bar f^*(\closednucl{K})=f^*(K)$ for any $K : \CompactOpens{A}$. Let $K : \CompactOpens{A}$ and note that
 \begin{align*}
   \bar f^*(\closednucl{K})
  &\quad=\quad  \bigvee_{K' : \CompactOpens{A}}  f^*\left( K\vee K'\right) \wedge \neg  f^*\left( K'\right)    \\
  &\quad=\quad   \bigvee_{K' : \CompactOpens{A}}  f^*\left( K\right) \wedge \neg  f^*\left( K'\right)  && [\text{$f^*$ preserves binary joins}]    \\
  &\quad= \quad    f^*(K)\wedge  \bigvee_{K' : \CompactOpens{A}}  \neg  f^*\left( K'\right)   && [\text{distributivity}] \\
  & \quad=\quad    f^*(K)\wedge  \One{X}&& [\Zero{X}\ \text{is compact}] \\
   & \quad=\quad  f^*(K),
  \end{align*}
  as required.
\end{proof}
\section{Summary and discussion}

We have constructed the patch locale of a spectral locale in the predicative and
constructive setting of univalent type theory. Furthermore, we have shown that
the patch functor \[\Patch : \Spec \rightarrow \Stone\] is the right adjoint to the
inclusion $\Stone \hookrightarrow \Spec$ which is to say that $\Patch$ exhibits the category
$\Stone$ as a coreflective subcategory of $\Spec$.
As we have elaborated on in Section~\ref{sec:spec-and-stone}, answering this
question in a predicative setting involves several new ingredients, compared to~\cite{mhe-patch-short,mhe-patch-full}:
\begin{enumerate}
  \item We have reformulated the notions of spectrality and Stone-ness in our
    predicative type-theoretic setting and have shown that crucial topological
    facts about these notions remain valid under these reformulations.
  \item We have developed several notions that capture the locale-theoretic
    notion of a basis and have shown their equivalences in
    Theorem~\ref{thm:answer} and in Theorem~\ref{thm:answer-stone}.
  \item In \MHELastName{}'s~\citeyearpar{mhe-patch-short} construction of the
    patch locale, the proof of the universal property relies on the existence of the
    frame of \emph{all} nuclei. As it is not clear that the poset of all nuclei
    can be shown to form a frame predicatively, we developed a new proof of the universal
    property using Lemma~\ref{lem:adjoint-to-induced-map}, which is completely
    independent of the existence of the frame of all nuclei.
\end{enumerate}

We have formalized all of our development, most importantly
Theorem~\ref{thm:patch-is-stone} and Theorem~\ref{thm:main}. The formalization
has been carried out by the third-named
author~\citep{type-topology-locale-theory} as part of the \VTypeTopology{}
library~\citep{type-topology}.

In previous work~\citep{mhe-patch-short, mhe-patch-full}, which forms the basis
of the present work, the patch construction was used to
\begin{enumerate}
  \item\label{item:coref-1}%
    exhibit $\Stone$ as a coreflective subcategory of $\Spec$, and
  \item\label{item:coref-2}%
    exhibit the category of compact regular locales and continuous maps as a
    coreflective subcategory of of stably compact locales and perfect maps.
\end{enumerate}
In our work, we have focused on item (\ref{item:coref-1}). The question of
taking a predicative approach to item (\ref{item:coref-2}) was previously
tackled by \cite{coq-zhang} using formal topology. We conjecture that it is
possible to instead use the approach we have presented here, namely, working
with locales with small bases and constructing the patch as the frame of
\VSCNuclei{}. It should also be possible to show predicatively that the category
of small distributive lattices is dually equivalent to the category of spectral
locales as defined in this paper.

\section*{Acknowledgement}

The first-named author acknowledges support from the Basque Government (grant
IT1483-22 and a postdoctoral fellowship of the Basque Government, grant
POS-2022-1-0015).

\bibliographystyle{msclike}
\bibliography{references}

\end{document}

%% file: macros.tex
\newcommand{\MHELastName}{Escard\'{o}}
\newcommand{\TDJLastName}{de~Jong}

\newcommand{\MHEName}{%
  \texorpdfstring{Mart\'in H.\ Escard\'o}{Martin\ H.\ Escardo}%
}

\newcommand{\IALastName}{Arrieta}
\newcommand{\IAName}{Igor\ \IALastName{}}

\newcommand{\ATLastName}{Tosun}
\newcommand{\ATName}{Ayberk\ \ATLastName{}}

\newcommand{\define}[1]{\emph{#1}}

\newcommand{\paren}[1]{\left(#1\right)}

\newcommand{\UU}{\mathcal{U}}
\newcommand{\VV}{\mathcal{V}}
\newcommand{\WW}{\mathcal{W}}

\newcommand{\USucc}[1]{{#1}^+}

\newcommand{\TOmit}[1]{}

\newcommand{\Spec}{\mathbf{Spec}}
\newcommand{\Stone}{\mathbf{Stone}}

\newcommand{\Patch}{\mathsf{Patch}}

\newcommand{\PatchTEXorPDF}{\texorpdfstring{$\Patch$}{\Patch}}

\newcommand{\is}{\vcentcolon\equiv}

\newcommand{\EmptyTy}{\mathbf{0}}
\newcommand{\UnitTy}{\mathbf{1}}
\newcommand{\NatTy}{\mathbb{N}}
\newcommand{\ListTy}[1]{\mathsf{List}(#1)}

\newcommand{\IsPropNm}{\mathsf{isProp}}
\newcommand{\IsProp}[1]{\IsPropNm\left(#1\right)}

\newcommand{\HProp}[1]{\Omega_{#1}}

\newcommand{\TruncTy}[1]{\left\| #1 \right\|}
\newcommand{\TruncTm}[1]{\left| #1 \right|}

\newcommand\ScaleExists[1]{\vcenter{\hbox{\scalefont{#1}$\exists$}}}

\DeclareMathOperator*\bigexists{%
  \vphantom\sum
  \mathchoice{\ScaleExists{1.95}}{\ScaleExists{1.55}}{\ScaleExists{1}}{\ScaleExists{0.9}}}

\newcommand{\FamNm}{\hyperref[defn:family]{{\color{black}\mathsf{Fam}}}}
\newcommand{\Fam}[2]{\FamNm_{#1}\left(#2\right)}

\newcommand{\Pair}[2]{(#1, #2)}

\newcommand{\EmbFam}[2]{\FamNm^{\hookrightarrow}_{#1}\left(#2\right)}

\newcommand{\Frm}{\mathbf{Frm}}

\newcommand{\Loc}{\mathbf{Loc}}

\newcommand{\opposite}[1]{{#1}^{\mathsf{op}}}

\newcommand{\opens}[1]{\mathcal{O}(#1)}

\DeclareMathOperator{\WayBelow}{\hyperref[defn:way-below]{{\color{black}\ll}}}

\DeclareMathOperator{\WellInside}{\prec}

\newcommand{\lub}{\sqcup}

\newcommand{\IdTy}[2]{#1 = #2}

\newcommand{\PiTySym}{\prod}
\newcommand{\PiTy}[3]{%
  \PiTySym_{#1 : #2}{#3}
}

\newcommand{\SigmaTypeSym}{\sum}
\newcommand{\projI}{\pi_1}
\newcommand{\projII}{\pi_2}

\newcommand{\image}[2]{\mathsf{image}(#1)}

\newcommand{\SigmaType}[3]{%
  \SigmaTypeSym_{#1 : #2}{#3}
}

\newcommand{\OneSym}{\mathbf{1}}
\newcommand{\One}[1]{\OneSym_{#1}}
\newcommand{\ZeroSym}{\mathbf{0}}
\newcommand{\Zero}[1]{\ZeroSym_{#1}}

\newcommand{\EquivSym}{\simeq}
\newcommand{\Equiv}[2]{#1 \simeq #2}

\newcommand{\FormalTopologyCitations}{%
  \citep{%
    int-formal-spaces, cssv-inductively-generated, coq-tosun-book-chapter%
  }%
}

\newcommand{\IsCompact}[1]{#1\ \text{\rm is compact}}

\newcommand{\closednucl}[1]{\mathbf{c}(#1)}
\newcommand{\opennucl}[1]{\mathbf{o}(#1)}

\newcommand{\VDirected}{\hyperref[defn:directed]{{\color{black}directed}}}

\newcommand{\VCompact}{{\color{black}compact}}

\newcommand{\VCompactOpen}{{\color{black}compact open}}

\newcommand{\VCompactLocale}{%
  \hyperref[defn:compact-locale]{{\color{black}compact locale}}
}

\newcommand{\VClopen}{clopen}

\newcommand{\VSpectralLocale}{%
  \hyperref[defn:spectral-locale]{{\color{black}spectral locale}}%
}

\newcommand{\CompactOpens}[1]{%
  {\color{black}\mathsf{K}}(#1)%
}

\newcommand{\Clopens}[1]{%
  \mathsf{C}(#1)
}

\newcommand{\VResizing}[2]{%
  $({#1}, {#2})${\color{black}-resizing}%
}

\newcommand{\VSmall}{\hyperref[defn:vsmall]{{\color{black}small}}}

\newcommand{\VVSmall}[1]{$#1$\hyperref[defn:vsmall]{{\color{black}-small}}}

\newcommand{\VHoTTUF}{\textsf{HoTT/UF}}

\newcommand{\RTwo}{{\MakeUppercase {\romannumeral 2}}}
\newcommand{\RSeven}{{\MakeUppercase {\romannumeral 7}}}

\newcommand{\VAgda}{\textsc{Agda}}

\newcommand{\VEmbedding}{\hyperref[defn:embedding]{{\color{black}embedding}}}

\newcommand{\VTypeTopology}{\textsc{TypeTopology}}

\newcommand{\VScottContinuous}{Scott continuous}

\newcommand{\VNucleus}{nucleus}
\newcommand{\VNuclei}{nuclei}

\newcommand{\VSCNucleus}{\VScottContinuous{} \VNucleus{}}
\newcommand{\VSCNuclei}{\VScottContinuous{} \VNuclei{}}

\newcommand{\VSplitSupport}{split support}

\newcommand{\VSpectral}{\hyperref[defn:spectral-locale]{{\color{black}spectral}}}

\newcommand{\VStone}{Stone}

\newcommand{\VWellInside}{well inside}

\newcommand{\FamEnum}[3]{(#3)_{#1 : #2}}

\newcommand{\emptyl}{\varepsilon}
\newcommand{\cons}{}
\newcommand{\append}{}

\newcommand{\ddnm}{k}
\newcommand{\dd}[1]{\ddnm_{#1}}